\documentclass[MPS,Times1COL]{WileyNJDv5} 

\articletype{Research Article}%

\received{Date Month Year}
\revised{Date Month Year}
\accepted{Date Month Year}
\journal{Journal}
\volume{00}
\copyyear{2023}
\startpage{1}

\usepackage{amsmath}
\usepackage{empheq} 
\usepackage{graphicx}

\usepackage{diagbox}
\usepackage{rotating}
\usepackage{tabularx}
\usepackage{underscore}
\usepackage{hyperref}
\hypersetup{hidelinks,
	colorlinks=true,
	allcolors=black,
	pdfstartview=Fit,
	breaklinks=true}

\usepackage[percent]{overpic}
\usepackage{amssymb}  
\usepackage{xcolor}
\usepackage{algorithm}
 \usepackage{algpseudocode}
\newcommand{\norm}[1]{\left\lVert #1 \right\rVert}

\newtheorem{ass}{Assumption}

\begin{document}

\title{Integrating Deep Learning and Contraction Theory for Robust Nonlinear State Estimation via Unsupervised Scientific Machine Learning}

\author[1,3]{Yasmine Marani}
\author[1,2]{Israel Filho}
\author[1]{Eric Feron}
\author[2]{Taous-Meriem Laleg-Kirati}

\authormark{Marani \textsc{et al.}}
\titlemark{Learning-based contraction nonlinear observer}

\address[1]{\orgdiv{Computer, Electrical and Mathematical Science and Engineering Division (CEMSE)}, \orgname{King Abdullah University of Science and Technology (KAUST)}, \orgaddress{\state{Thuwal}, \country{Saudi Arabia}}}

\address[2]{ \orgname{Université Paris-Saclay, Inria, CIAMS, Gif-sur-Yvette, 91190}, \orgaddress{\state{ Paris-Saclay}, \country{France}}}

\address[3]{ \orgname{Division of Decision and Control Systems, KTH Royal Institute of Technology, SE-100 44}, \orgaddress{\state{Stockholm}, \country{Sweden}}}

\corres{Division of Decision and Control Systems, KTH Royal Institute of Technology, SE-100 44, Stockholm, Sweden. \email{yasmine.marani@kaust.edu.sa, ymarani@kth.se}}


\abstract[Abstract]{A common way to design observers is to add a correction term to a copy of the system; however, designing the correction term for nonlinear systems remains a significant long-standing challenge. Contraction theory offers a unified approach to designing this correction term by solving a matrix partial differential inequality (MPDI) and identifying a contraction metric. However, solving the MPDI for both the correction term and the contraction metric is highly challenging, both analytically and numerically. Therefore, the aim of this paper is to propose a learning-based approach to determine both the observer's correction term and the contraction metric by incorporating the contraction requirements into the learning process. The proposed approach relies on a scientific machine learning formulation that embeds the contraction conditions into the training loss function.  The proposed approach is then extended to non-autonomous systems while keeping the correction term and contraction metric static to avoid generalization issues arising from time-dependent training. Computable bounds on the learning errors of the proposed observer are established as a function of the training residual and the sampling resolution. Furthermore, the robustness of the proposed observer to measurement noise and learning errors are established 
in an exponential input-to-state stability sense. Based on the robustness analysis, the present paper takes a further step by proposing a robust learning-based contraction nonlinear observer.
The proposed observers are evaluated in numerical simulations for different contraction rates and measurement noise levels.}



\maketitle



\section{Introduction}
\label{sec:intro}

\noindent
Nonlinear observer design is a core research area in control theory that continues to receive sustained interest. Although general and systematic methods for the state estimation of linear systems with global convergence guarantees are well-established in the literature \cite{Luenberger1964,Kalman1960}, the design of observers for nonlinear systems still lacks comparable generality and global convergence guarantees.\\
A wide range of nonlinear observers has been proposed in the literature,
including high-gain observers \cite{khalil2008}, immersion and invariance-based observers \cite{karagiannis2008invariant}, geometric
observers \cite{KRENER1983}, LMI-based observers \cite{Rajamani1998}, algebraic estimators \cite{Yasmine2023}, latent-space immersion methods
\cite{Kazantzis1997}, and the widely used Extended Kalman Filter (EKF) \cite{anderson1979}. Most of these, however, either apply only to specific
classes of nonlinearities or, when applicable to general systems, guarantee only local convergence. Contraction-based observer design has recently regained attention as an alternative that avoids both limitations, offering a unified
framework for a broad class of smooth nonlinear systems together with global
exponential convergence.\\
 Initiated more than a century ago by S. Banach in \cite{Banach1922}, contraction theory received significant attention for fixed-point problems, and only a few works considered control systems \cite{Demidovic1961,Krasovskii1963,boyd1994linear}. It was not until the late 90s that contraction theory for the analysis of dynamical systems was fully formalized in \cite{LOHMILLER1998683}. Unlike traditional methods, such as Lyapunov theory \cite{Nijmeijer1990, isidori1995nonlinear,  khalil2002nonlinear}, contraction theory does not establish the stability of nonlinear systems with respect to an equilibrium point or nominal trajectory. Instead, it determines whether any two trajectories of a given nonlinear system converge to each other. In other words, a system is contracting if its final behavior is independent of the initial condition, which is a key feature for observer design \cite{manchester2018contracting}. The reader is referred to \cite{TSUKAMOTO2021135}  and \cite{Bullo} for more comprehensive reviews of contraction theory. 
 A wide range of nonlinear controllers and observers based on contraction theory have been introduced in the literature, including those in \cite{slotine1996a, slotine1996b, lohmiller1997, LOHMILLER1998683, sanfelice2011convergence, dani2014observer, manchester2018contracting}. In particular, contraction theory offers a general framework for designing nonlinear observers with guaranteed global exponential convergence. The design approach relies on solving a matrix partial differential inequality (MPDI) to find the observer's correction term and the contraction metric. Nevertheless, despite the theoretical appeal of this approach, solving the MPDI poses a significant analytical and numerical challenge, which constrains the practical application and implementation of contraction-based observers with nonlinear correction terms \cite{PeterGiesl2023}. Therefore, the present paper aims to address this challenge by proposing a scientific learning-based approach to design the observer's correction term and find the contraction metric by promoting the MPDI in the learning process. \\
\noindent 
The paradigm of leveraging Artificial Neural Networks (ANNs) to approximate solutions of differential equations has gained significant traction. A prominent methodology within this domain is the Physics-Informed Neural Network (PINN) initially proposed by Raissi \textit{et al.}~\cite{raissi2019}. This approach embeds the governing physics laws, expressed as ordinary or partial differential equations (ODEs/PDEs), directly into the network's training objective as a regularization term. The use of automatic differentiation is central to this process, enabling the efficient computation of the derivatives required to enforce the physics constraints within the loss function.
The versatility of PINNs has led to their extensive application in modeling, parameter estimation, and observer design for nonlinear dynamical systems~\cite{zhai2023parameter, antonelo2024physics, ekeland2024physics}. These methods have proven effective by ensuring that the learned models are not only data-efficient but also physically consistent, achieving robust generalization. Notably, in the context of observer design, PINN-based frameworks have been developed for both discrete and continuous-time systems~\cite{alvarez2024nonlinear}, including variations for designing Kazantzis-Kravaris-Luenberger (KKL) observer~\cite{niazi2023learning, peralez2021deep, Yasmine2023kkl, MARANI2025}. A common feature in many of these applications is the explicit dependence of the neural network on time as an input variable, which is instrumental when learning system trajectories or time-varying parameters. In this work, we adopt a different perspective inspired by PINNs, which leads us to believe that we are in the realm of scientific machine learning. While our method does not explicitly use time as a network's input, it explores the inherent plasticity of neural network architectures. Specifically, we parameterize the contraction metric directly as a learnable component of the neural network, which is then learned during the training process to shape the observer's correction term.\\ 
\noindent
In our previous work \cite{MaraniIsrael2025}, we proposed an unsupervised
learning approach to learn the observer correction term for a fixed identity
(Euclidean) metric. In the present paper, we remove the conservatism of this
fixed choice by treating the contraction metric as an unknown symmetric positive definite matrix to be learned jointly with the correction term.
The main contributions of this paper are summarized as follows:
\begin{itemize}
  \item We propose an unsupervised scientific machine learning approach that simultaneously learns the observer's correction term and the
        contraction metric by embedding the contraction conditions directly into the
        training loss.
    \item We extend the proposed observer to non-autonomous systems using a static correction term and contraction metric, avoiding time-dependent PINNs that limit generalization beyond the training horizon.
    \item Derive computable bounds on the learning errors 
    \item Establish the robustness properties of the proposed observer in the presence of learning errors, process and measurement noise, yielding an E-ISS result under sufficient conditions of the learning error and the spectrum of the contraction metric. 
    \item Propose a robust learning-based contraction observer with disturbance rejection and noise attenuation properties by further optimizing the spectrum of the contraction metric.
  \item Validate the proposed observers on four academic nonlinear systems across a range of contraction rates and measurement noise levels.
\end{itemize}

\noindent
The rest of the paper is organized as follows. Contraction-based nonlinear observer for autonomous systems and the problem formulation are presented in Section \ref{sec:prelim}. The main results on the unsupervised learning-based contraction nonlinear observer with a learnable contraction metric are stated in section \ref{sec:pinnDes}. Afterwards, the proposed observer is extended in Section \ref{sec:non-auto} to non-autonomous systems, while keeping the correction term and the contraction metric static. Section \ref{sec:robustnessAnalysis} establishes computable bounds on the learning errors, followed by robustness properties to process and measurement noise in the presence of learning errors. Following the robustness analysis, a robust learning-based contraction observer is proposed in Section \ref{sec:robustnessAnalysis}. The performance and robustness of the proposed observers are assessed through numerical simulations in section \ref{sec:sim} for different levels of measurement noise and contraction rates. Finally, a summary of the contributions and future work directions are provided in Section \ref{sec:conclusion}. \\

\noindent
\textbf{Notation.} 
For a square matrix $M$, $\operatorname{He}\{M\}=\frac{1}{2}(M+M^T)$ is the Hermitian part of the matrix $M$, $\|M\|$ is its induced 2-norm, $\lambda(M)$ its eigenvalues, and $\kappa$. For a symmetric matrix $P$, $\lambda_{\min}(P)$ and $\lambda_{\max}(P)$, respectively, represent the smallest and largest eigenvalue of $P$. The class $\mathcal{C}^k$ is the class of $k$ continuously differentiable functions. A $\mathcal{KL}$ function $\rho(r,s):[0; +\infty)\times [0;+\infty) \rightarrow [0;+\infty)$ is a continuous function who is strictly increasing in $r$, with $\rho(0,.)=0$, and decreasing in $s$, such that for $s\rightarrow +\infty$, $\rho(., s)\rightarrow 0$. The Euclidean norm of a vector $u$ is denoted by $\norm{u}$, and its infinity norm is denoted by $\|u\|_{\infty}$. Let $\Omega$ be a non-empty compact set, $\mu(\Omega)$ is its Lebesgue measure. Let $\mathcal{D}$ be a finite set, $|\mathcal{D}|$ denotes its cardinality. Let $\lceil . \rceil$ denotes the ceil function.

\section{Preliminaries and Problem Statement}\label{sec:prelim}
\subsection{Background on contraction theory for nonlinear observer design}

Contraction theory offers a new perspective on the stability analysis of nonlinear systems. Unlike the classical Lyapunov theory, the stability of a given system through the lens of contraction theory is established independently of the knowledge of any equilibrium point or nominal trajectory. A given system is said to be \textit{contracting} if it \textit{"forgets"} its initial condition and all the trajectories of the solutions converge to each other \cite{LOHMILLER1998683}. All of which are considered to be desirable properties for observer design. In particular, contraction theory offers a universal approach to designing nonlinear observers \cite{slotine1996a,slotine1996b}. \\   
Consider the following autonomous nonlinear system.
\begin{equation}
    \label{model}
    \left\{\begin{array}{l}\dot{x}(t)=f(x)\\
    y(t)=h(x),
    \end{array}\right.
\end{equation}
where $x(t) \in \mathbb{R}^n$ is the state,  $y(t) \in \mathbb{R}^p$ is the output,  $f: \mathbb{R}^n  \rightarrow \mathbb{R}^n$ and $h: \mathbb{R}^n \rightarrow \mathbb{R}^p$ are smooth vector fields.\\
Consider the associated observer constructed by taking a copy of the dynamic and adding a nonlinear correction term: 
\begin{equation}
\label{observer}
\left\{\begin{array}{l}
\dot{\hat{x}}=f(\hat{x})+k(\hat{x}, y) \\
\hat{y}=h(\hat{x}),
\end{array}\right.
\end{equation}
where $\hat{x} \in \mathbb{R}^n$ is the estimated state,  $\hat{y} \in \mathbb{R}^p$ is the estimated output, and $k: \mathbb{R}^n \times \mathbb{R}^p \rightarrow \mathbb{R}^n $ is the observer's nonlinear correction term. \\
 Designing the nonlinear correction term to ensure the convergence of the estimation error to the origin is the main challenge in nonlinear observer design. Contraction theory offers a general approach to design a smooth nonlinear correction term, which is provided by the following theorem.   

\begin{theorem} \cite{BERNARD2022}
\label{th1}
Consider the smooth nonlinear system \eqref{model} with $f$ of at least class $\mathcal{C}^1$. If there exists a symmetric and positive definite matrix $P\in \mathbb{R}^{n\times n}$, a $\mathcal{C}^1$ function $k: \mathbb{R}^n \times \mathbb{R}^p \rightarrow \mathbb{R}^n $ and a real positive number $c$, such that 
\begin{equation}\label{contractth}
\left\{\begin{aligned}
& \operatorname{He}\left\{P\left[\frac{\partial f}{\partial \hat{x}}(\hat{x})+\frac{\partial k}{\partial \hat{x}}(\hat{x}, y)\right]\right\} \leqslant-c P
 \quad \forall(\hat{x}, y) \in \mathbb{R}^n \times \mathbb{R}^p\\
 &k(x, h(x))=0 \quad \forall x \in \mathbb{R}^n,  \\
\end{aligned} \right.
\end{equation}
Then the observer in \eqref{observer} is a globally exponentially stable observer for system \eqref{model}. 
\end{theorem}

\subsection{Problem Statement}
Solving the infinite-dimensional inequality for $k$ and $P$ is a very challenging problem both analytically and numerically. In fact, there are no numerical solvers, according to the best of our knowledge, that are equipped to solve matrix partial differential inequalities. Relying on the universal approximation principle and the power of neural networks to approximate complex nonlinear functions, as well as the automatic differentiation capabilities of physics-informed neural networks, we proposed in \cite{MaraniIsrael2025} an unsupervised physics-informed neural network to learn the correction term by fixing the contraction metric $P$ to identity. In the present paper, we aim to lift any restrictions that result from pre-fixing $P$,  by proposing a Physics-Informed Neural Networks (PINNs) that simultaneously learns the observer's correction term and the contraction metric. The proposed learning approach operates in an unsupervised manner, guided by the system's underlying dynamics and the contraction conditions of  Theorem \ref{th1}, rather than by pre-calculated trajectory data.\\
Since relying on learning approaches requires training on a compact set of interests, we consider in the remainder of the present paper autonomous nonlinear systems \eqref{model} satisfying the assumption below, which is a relaxation of the assumption used in \cite{MaraniIsrael2025}.
\begin{ass}\label{asslearning}  
The solutions of interest of system \eqref{model} initialized in $\mathcal{X}_0$ remain in a compact convex set $\mathcal{X}$, with $\mathcal{X}_0\subseteq\mathcal{X}$, and the associated outputs remain in $\mathcal{Y}$.
\end{ass}

\noindent
Note that the convexity of $\mathcal{X}$ is not a requirement for the training but for ensuring the exponential convergence of the observer in \eqref{observer} on $\mathcal{X}$ \cite{Bullo}. \\
Taking the above into consideration, we aim in this paper to solve the following 
\begin{equation}\label{eq:ideal1}
\left\{ \begin{aligned}
  \text{find} \quad & P = P^T \succ 0, \quad k \in \mathcal{C}^1 \\
  \text{s.t} \quad
  & \operatorname{He}\!\left\{P\!\left[
      \tfrac{\partial f}{\partial\hat{x}}(\hat{x})
      + \tfrac{\partial k}{\partial\hat{x}}(\hat{x},y)
    \right]\right\} + cP \preceq 0,
    \quad \forall(\hat{x}, y) \in \mathcal{X} \times \mathcal{Y},\\
  & k\!\left(x,\,h(x)\right) = 0,
    \quad \forall\,x \in \mathcal{X},\\
\end{aligned}\right.
\end{equation}

\section{Learning-based Contraction Observer for autonomous nonlinear systems}
\label{sec:pinnDes}

\begin{figure}
    \centering
    \includegraphics[width=0.7\linewidth]{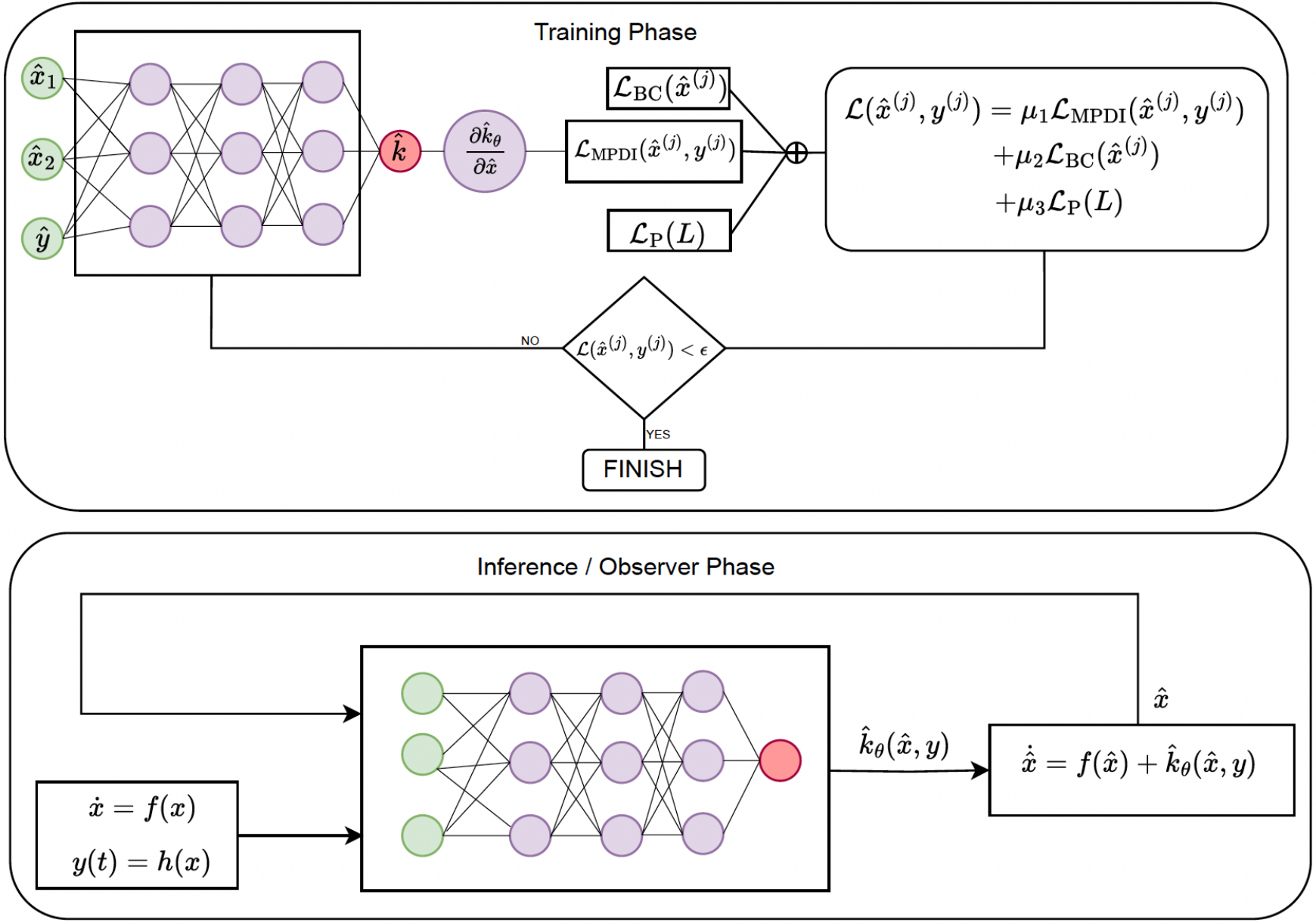}
    \caption{Block diagram of the proposed unsupervised learning-based contraction nonlinear observer design. The first block depicts the training phase that computes the nonlinear correction term $\hat{k}_{\theta}$ and the contraction metric $\hat{P}$. The second block is a high-level overview of the operation mode of the proposed observer.}
    \label{fig:systemVisual}
\end{figure}

\noindent
This section details the proposed learning framework for designing the learning-based contraction observer.
Advancing upon the research in \cite{MaraniIsrael2025}, we propose the simultaneous learning of both the observer's correction term and the contraction metric. In our proposed approach, the contraction metric is treated as a learnable parameter within the neural network, optimized concurrently with the network that defines the correction term, thus moving beyond the conservative assumption of a Euclidean metric ($P=I$).\\
\noindent
Solving the problem in \eqref{eq:ideal1} is relaxed into the following minimization problem
\begin{equation}
    \min_{\theta,L} \mathcal{L}(\theta, L)= \mu_1 \mathcal{L}_{\text{MPDI}}(\theta,L)
+ \mu_2 \mathcal{L}_{\text{BC}}(\theta)
+ \mu_3 \mathcal{L}_{P}(L)
\label{eqoptimize}
\end{equation}
where $\mathcal{L}_{\text{MPDI}}$, $\mathcal{L}_{\text{BC}}$ and $\mathcal{L}_{\text{P}}$ are the loss functions promoting the contraction conditions of Theorem \ref{th1}, and $\mu_1, \mu_2,$ $\mu_3$ are tunable coefficients balancing the influence of each component.  The parameters  $\theta$ and $L$ are the neural network's learnable parameters for the correction term and contraction metric, respectively.

\noindent Relying on the universal approximation theorem, we consider the following feedforward parametrization of the observer correction term $\hat{k}_\theta: \mathbb{R}^n \times \mathbb{R}^p \rightarrow \mathbb{R}^n$ solution to the optimization problem in \eqref{eqoptimize}:
\begin{equation}
\label{eq:mlpCompMath}
\hat{k}_{\theta}(\hat{x},y) = W_H Z_{H-1} \circ \dots \circ Z_1(Z_0) + b_H,
\end{equation}
where $Z_i(Z_{i-1})=\sigma(W_iZ_{i-1}+b_i)$, is the output of each layer $i$, for $i=1,...,H$, where $H$ is the number of layers, $Z_0=\text{vec}(\hat{x},y)$ the input of the neural network consisting of a vertical concatenation of $\hat{x}$ and $y$, and $\sigma$ is the activation function. The neural network's parameters tensor is denoted by $\theta=\{W_i, b_i\}_{i=1}^H$, where $W_i$ are weights matrices and $b_i$ are biases vectors of appropriate dimensions.\\
Furthermore, we propose to parametrize the contraction metric by a  Cholesky-like decomposition, 
\begin{equation}
    \hat{P} = L L^\top, 
    \label{eq:p}
\end{equation}
where $\hat{P}$ is the solution to the optimization problem in \eqref{eqoptimize} and $L\in\mathbb{R}^n$ is a lower triangular matrix whose elements are learnable parameters.
This parametrization ensures that the learned contraction metric is at least symmetric. Further conditions ensuring the positive definiteness of $\hat{P}$ will be discussed in the subsections below.\\

\noindent
As the training does not rely on measured data, the physics-based loss in \eqref{eqoptimize} is minimized over a set of $N$ collocation points $\mathcal{D} = \{\hat{x}^{[j]}, {y}^{[j]}\}_{j=1}^{N}$ constructed by uniformly sampling from the state and output sets of interests $(\mathcal{X} \times \mathcal{Y})$ following Assumption \ref{asslearning}. \\
The block diagram of the proposed learning-based nonlinear observer is depicted in FIGURE.~\ref{fig:systemVisual}, and details of the construction of the composite loss function in \eqref{eqoptimize} for promoting the contraction conditions of Theorem \ref{th1} are provided below.

\noindent

\subsubsection*{MPDI Loss ($\mathcal{L}_{\text{MPDI}}$)}
The primary loss term enforces the Matrix Partial Differential Inequality (MPDI) from the contraction conditions in \eqref{contractth}:
\begin{equation}
    \label{eq:contractionInequalityV1}
    D(\hat{x},y) \triangleq \operatorname{He}\left\{ \hat{P} \left[\frac{\partial f}{\partial \hat{x}}(\hat{x})+\frac{\partial \hat{k}_\theta}{\partial \hat{x}}(\hat{x}, y) \right]\right\} + c \hat{P} \preceq 0.
\end{equation}
To penalize violations of this negative semi-definiteness condition, we propose the following loss function
\begin{equation}
\label{eq:minorsV1}
\mathcal{L}_{\text{MPDI}} = \frac{1}{N} \sum_{j=1}^{N}\sum_{k=1}^{n}\frac{1}{\beta} \ln\left(1 + \exp{\left(\beta \lambda_i \left(D(\hat{x}^{[j]},{y}^{[j]})\right)\right)}\right) ,  
\end{equation}

\noindent
where $\lambda_i$ is the $i$-th eigenvalue of $D(\hat{x}^{[j]},{y}^{[j]})$, for $i=1,\ldots,n$, and $\beta$ is a large positive constant.
Unlike the loss in \cite{MaraniIsrael2025}, we replace the ReLU function in \eqref{eq:minorsV1} with a smooth approximation, thereby avoiding the issues arising from its non-differentiability at zero.
To this end, we adopt the Softplus function, $\operatorname{Softplus}(\cdot)\triangleq \ln\left(1 + \exp(\cdot)\right)$, a widely used activation function in deep neural networks \cite{zheng2015improving}.
Indeed, as $\beta \rightarrow +\infty$, we have $\frac{1}{\beta}\operatorname{Softplus}(\beta a) \rightarrow \operatorname{ReLU}(a)$.
  
\subsubsection*{Boundary Condition Loss ($\mathcal{L}_{\text{BC}}$)}
This component enforces the logical constraint that the observer correction term $\hat{k}_\theta$ should vanish when the estimated output aligns with the system's output (i.e., when $y = h(\hat{x})$). This is achieved by penalizing the correction term's norm under this condition:
\begin{equation}
\mathcal{L}_{\text{BC}} = \frac{1}{N} \sum_{j=1}^{N} \left\|\hat{k}_\theta(\hat{x}^{[j]},h(\hat{x}^{[j]}))\right\|^2.  
\label{eq:eqBCV1}
\end{equation}

\subsubsection*{Contraction Metric Regularization Loss ($\mathcal{L}_P$)}
Exploiting the Cholesky-like decomposition of the contraction metric, the positive definiteness of $\hat{P}$ is obtained if and only if the diagonal entries of $L$ are non-zero.
To ensure that $L$ is well-behaved and non-singular, we add a regularization term that penalizes trivial diagonal elements
\begin{equation}
\label{eq:LdiagLossV1}
\mathcal{L}_{P}=\sum_{k=1}^{n}\frac{1}{l_{ii}^2+\epsilon
},
\end{equation}
where $l_{ii}$ are the diagonal elements of the learnable matrix $L$, for $i=1,...,n$, and $\epsilon$ is a small constant added to avoid division by zero. \\

\noindent
The training pseudo-code for the proposed learning-based contraction nonlinear observer is provided in Algorithm 1.\\

\begin{algorithm}
\label{Algo}
\caption{Training Algorithm of the learning-based Contraction Nonlinear Observer with Learnable contraction Metric}
\begin{algorithmic}[1]
\State \textbf{Inputs:} dataset $\mathcal{D}$; initial parameters $\Theta=(\theta,L)$; nonlinear function $f$ and $h$; contraction rate $c>0$; Adam learning rate $\alpha$;  epochs $N_{\text{Adam}},N_{\text{LBFGS}}$; loss weights $(\mu_1,\mu_2,\mu_3)$;  regularization constant $\epsilon$; softplus sharpness $\beta>0$; batch size $|\mathcal{B}|$.
\For{$q=1,\dots,N_{\text{Adam}}$} \Comment{Adam phase}
  \For{mini-batch $\mathcal{B}\subset\mathcal{D}$}
    \State $P \gets L L^\top$
    \State Compute $\hat{k}_\theta(\hat{x},y)$ on $\mathcal{B}$
    \State Compute the jacobians $\big\{\frac{\partial f}{\partial \hat{x}}(\hat{x}),\,\frac{\partial \hat{k}_\theta}{\partial \hat{x}}(\hat{x},y)\big\}$ on $\mathcal{B}$
    \State Build $D(\hat{x},y) = \operatorname{He}\!\Big\{ P\big[\tfrac{\partial f}{\partial \hat{x}}+\tfrac{\partial \hat{k}_\theta}{\partial \hat{x}}\big]\Big\} + cP$
    \State Compute the eigenvalues $\{\lambda_i(D(\hat{x}^{[j]},y^{[j]}))\}_{i=1}^{n}$ on $\mathcal{B}$
    \State $\mathcal{L}_{\text{MPDI}} \gets \frac{1}{|\mathcal{B}|}\sum_{j}\sum_{i=1}^{n}\frac{1}{\beta}\ln\!\big(1+\exp(\beta\,\lambda_i(D(\hat{x}^{[j]},y^{[j]})))\big)$ \Comment{Eq.~\eqref{eq:minorsV1}}
    \State $\mathcal{L}_{\text{BC}} \gets \frac{1}{|\mathcal{B}|}\sum_{j}\|\hat{k}_\theta(\hat{x}^{[j]},h(\hat{x}^{[j]}))\|^2$ \Comment{Eq.~\eqref{eq:eqBCV1}}
    \State $\mathcal{L}_{P} \gets \sum_{k=1}^{n}\frac{1}{l_{ii}^2+\epsilon}$ \Comment{Eq.~\eqref{eq:LdiagLossV1}}
    \State $\mathcal{L}(\theta,L) \gets \mu_1\mathcal{L}_{\text{MPDI}}+\mu_2\mathcal{L}_{\text{BC}}+\mu_3\mathcal{L}_{P}$ \Comment{Eq.~\eqref{eqoptimize}}
    \State $\Theta \gets \Theta - \alpha\,\nabla_{\Theta}\mathcal{L}(\theta,L)$ \Comment{jointly update $\theta$ and $L$}
  \EndFor
\EndFor
\If{$N_{\text{LBFGS}}>0$} \Comment{optional refinement}
  \For{$q=1,\dots,N_{\text{LBFGS}}$}
    \State $P \gets L L^\top$
    \State Recompute $\mathcal{L}(\theta,L)$ via Eq.~\eqref{eqoptimize} (full data or large batch)
    \State $\Theta \gets \text{LBFGS\_step}\big(\Theta,\nabla_{\Theta}\mathcal{L}\big)$
  \EndFor
\EndIf
\State \textbf{Output:} optimized parameters $\Theta^\star=(\theta^\star,L^\star)$.
\end{algorithmic}
\end{algorithm}

\section{Learning-based Contraction Observer for non-autonomous nonlinear systems with static correction term and contraction metric}
\label{sec:non-auto}
In the present section, we consider a non-autonomous  nonlinear system with 
\begin{equation}
    \label{modelinput}
    \left\{\begin{array}{l}\dot{x}(t)=f(x,u)\\
    y(t)=h(x),
    \end{array}\right.
\end{equation}
where $u(t)\in \mathcal{U}\subset\mathbb{R}^m$ is a time-varying input.\\
Ideally, designing a contraction nonlinear observer for \eqref{modelinput} requires solving \eqref{contractth} for the time-varying correction term and contraction metric. However, this setup is not ideal from a learning and implementation point of view. Indeed, having the time as an input for the neural network raises several concerns:
\begin{enumerate}[(i).]
    \item Training requires sampling from a fixed time horizon $[0; T]$, which affects the generalization of the PINN for $t> T$.
    \item The resulting learning-based contraction observer cannot be implemented at a finer time resolution than that of the time grid used during training.
\end{enumerate}
In order to bypass these limitations, we aim in this section to design a learning-based contraction nonlinear observer for system \eqref{modelinput} using a static correction term and contraction metric.\\
Notice that the non-autonomous system in \eqref{modelinput} is equivalent to its autonomous counterpart in \eqref{model} for a constant input $u(t)=\bar{u}=cte$. Based on this observation, we provide conditions under which the following observer 
\begin{equation}
\label{observeru}
\left\{\begin{array}{l}
\dot{\hat{x}}=f(\hat{x},u)+k(\hat{x}, y)\\
\hat{y}=h(\hat{x}).
\end{array}\right.
\end{equation}
with a correction term learned from the autonomous version of \eqref{modelinput} for $u=\bar{u}=$cte provides an exponential convergence. To this end, let us define the following nonlinear function 
\begin{equation}
    \Phi(x,u) \triangleq f(x,u)-f(x,\bar{u}),
    \label{phi}
\end{equation}
and consider the assumptions below.
\begin{ass}\label{assu}
For any  input of interest $u \in \mathcal{U}$, the solutions of interests of non-autonomous system \eqref{modelinput} initialized in $\mathcal{X}_0$, remain in the compact convex set $\mathcal{X}$, with $\mathcal{X}_0\subseteq\mathcal{X}$ and $ h(X(t,x_0,u(t)) \in \mathcal{Y}$, where $X(t,x_0,u(t))$ is the solution of \eqref{modelinput} initialized at $x_0$ and the sets $\mathcal{X}_0, \mathcal{X}, \mathcal{Y}$ are as in Assumption \ref{asslearning}.
\end{ass}
\begin{ass}
    The nonlinear function $\Phi$ is Lipschitz in $x$ uniformly in $u$, i.e.,  $\norm{ \Phi(x_1, u)-\Phi(x_2, u)}\leq  l_{\Phi}\norm{x_1-x_2},  \forall u\in \mathcal{U}.$
, with $l_{\Phi}>0$.
\label{lipPhi}
\end{ass}
\begin{remark}
Assumption~\ref{lipPhi} is a mild condition. A sufficient condition for it to hold is that $f(\cdot, u)$ is Lipschitz continuous in its first argument, uniformly in $u$, which is a standard assumption for guaranteeing the uniqueness of solutions. It is worth noting, however, that Lipschitz continuity of $f$ is not necessary for Assumption~\ref{lipPhi} to be satisfied.
\end{remark}
\begin{theorem}\label{thu}
Consider system \eqref{modelinput} satisfying Assumption \ref{assu}, and the observer in \eqref{observeru}  with a static correction term $k$ and contraction metric $P$ satisfying 
\begin{equation}\label{eq:idealu}
\left\{ \begin{aligned}
  & \operatorname{He}\!\left\{P\!\left[
      \tfrac{\partial f}{\partial\hat{x}}(\hat{x}, \bar{u})
      + \tfrac{\partial k}{\partial\hat{x}}(\hat{x},y)
    \right]\right\} + cP \preceq 0,
    \quad \forall(\hat{x}, y) \in \mathcal{X} \times \mathcal{Y},\\
  & k\!\left(x,\,h(x)\right) = 0,
    \quad \forall\,x \in \mathcal{X},\\
\end{aligned}\right.
\end{equation}
for a contraction rate $c>0$. Furthermore, let Assumption \ref{lipPhi} hold. If $l_{\Phi}<\frac{c \lambda_{min}(P)}{\lambda_{max}(P)}$, then the estimation error convergences exponentially and satisfies  
\begin{equation}
\|x(t)-\hat{x}(t)\| \leqslant  \eta _0\|x(0)-\hat{x}(0)\| \exp{\left(-\left(c-l_{\Phi}\frac{\lambda_{max}(P)}{\lambda_{min}(P)}\right) t\right)}.\\
\label{estimationerroru}
\end{equation}
\noindent
\end{theorem}

\begin{proof}
Consider the estimation error $e=x-\hat{x}$
and the quadratic Lyapunov function $V=\frac{1}{2}e^TPe$. Taking the time derivative yields
\begin{align}
    \dot{V}&=e^TP\left\{f(x,u)-f(\hat{x},u) -k(\hat{x},y)  \right\} \notag\\
    & = e^TP\left\{f(x,\bar{u})-f(\hat{x},\bar{u}) -k(\hat{x},y)  \right\} + e^TP\left\{\Phi(x,u)-\Phi(\hat{x},u)\right\}\notag
\end{align}
Considering Assumption \ref{assu} with $k\!\left(x,\,y=h(x)\right) = 0$ for all $x \in \mathcal{X}$, with $\mathcal{X}$ convex, then by using the fundamental theorem of calculus as in \cite{Bullo, BERNARD2022} one gets\\
$$\dot{V}=e^TP \left( \int_0^1 \frac{\partial f}{\partial x}(\hat{x} + \alpha(x - \hat{x}), \bar{u}) + \frac{\partial k}{\partial x}(x + \alpha(\hat{x} - x), y) \, d\alpha \right) \times e+ e^TP\left\{\Phi(x,u)-\Phi(\hat{x},u)\right\}$$
The contraction inequality in \eqref{eq:idealu} and Assumption \ref{lipPhi} yield
\begin{align}
\dot{V} &\leq -ce^TPe + l_{\Phi} \lambda_{max}(P) e^Te \notag \\
& \leq -2\left( c-l_{\Phi} \frac{\lambda_{max}(P)}{\lambda_{min}(P)}\right) V
\end{align}
Finally, if $l_{\Phi}<\frac{c \lambda_{min}(P)}{\lambda_{max}(P)}$, then the estimation error converges exponentially to the origin and satisfies \eqref{estimationerroru}.
\end{proof}

\section{Robustness analysis and disturbance rejection of the PINN-based nonlinear contraction observer to learning errors }
\label{sec:robustnessAnalysis}

\subsection{Learning error bound guarantees}
Solving problem \eqref{eq:ideal1} does not lead to a unique solution due to the contraction inequality, i.e.,  any correction term and contraction metric that satisfy this inequality are equally optimal. However, one should distinguish the solutions of \eqref{eq:ideal1} from  the minimizers of the loss in \eqref{eqoptimize} for the following reasons:
\begin{enumerate}[(i).]
\item Unlike the problem in \eqref{eq:ideal1}, the loss function is minimized over a finite cardinality dataset $\mathcal{D}\subset \mathcal{X}\times \mathcal{Y}$.\\
\item In practice, training a neural network does not drive the loss to zero. Consequently, the learned correction term $\hat{k}_{\theta}$ and contraction metric $\hat{P}$ satisfy the contraction conditions up to a small residual error.
\end{enumerate}
In the present subsection, we aim to provide deterministic and computable learning generalization error bounds that will be used in further robustness analysis in the upcoming subsection. If the reader is interested in probabilistic bounds, one can refer to the work in \cite{mohri2018foundations}. \\
In line with this, we consider a validation dataset $\mathcal{V}$ constructed by an equidistant sampling of $\mathcal{X}\times\mathcal{Y}$ as in Definition \ref{defA3} in Appendix A, with a resolution $\rho(\mu,N_v)=\left(\frac{\mu(\mathcal{R})}{N_v}\right)^{\frac{1}{n+p}}$, where $\mu$ is the Lebesgue measure of the smallest Axis-Aligned Bounding Orthotope (AABO) $\mathcal{R}$ enclosing $\mathcal{X}\times\mathcal{Y}$, and $N_v$ is a chosen natural number large enough.\\
Let the learned correction term and contraction metric satisfy the following on $\mathcal{V}$.

\begin{subequations}\label{residual}
\renewcommand{\theequation}{\theparentequation.\alph{equation}}
\begin{empheq}[left=\empheqlbrace]{align}
&\mathrm{He}\left\{ \hat{P}\left[ \tfrac{\partial f}{\partial\hat{x}}(\hat{x}^{[j]}) +  \tfrac{\partial \hat{k}_{\theta}}{\partial\hat{x}}(\hat{x}^{[j]},y^{[j]})\right] \right\} + c\hat{P} \leq \delta I,
   && \forall (\hat{x}^{[j]},y^{[j]})\in \mathcal{V} \label{residual:a}\\
& \hat{k}_{\theta}(x^{[j]}, h(x^{[j]})) = \varepsilon(x^{[j]}),
   && \forall x^{[j]} \in \mathcal{V}_x  \label{residual:b}
\end{empheq}
\end{subequations}

\noindent
where $\delta$ is the MPDI residual, $\varepsilon$ is the boundary condition residual, and $\mathcal{V}_x=\{x^{[j]}\}_{j=1}^{j=N_v}\subset \mathcal{V}$.\\

\noindent
Consider the following set of assumptions:

\begin{ass}
\label{ass:learning}
    There exist finite positive constants $\bar{\varepsilon}$, $\kappa_{1x}$, $\kappa_{2}$, and $l_f$ such that 
    \begin{itemize}
        \item (A\ref{ass:learning}.1): $\|\varepsilon(x^{[j]})\|\leqslant \bar{\varepsilon} \quad \forall x^{[j]}\in \mathcal{V}_x$
        \item (A\ref{ass:learning}.2):$\|\hat{k}_\theta(x_1,h(x_1))-\hat{k}_\theta(x_2,h(x_2))\|\leqslant \kappa_{1x} \|x_1-x_2\| \quad \forall x_1,x_2\in \mathcal{X}^2$
    \item (A\ref{ass:learning}.3):$\| \frac{\partial \hat{k}_{\theta}}{\partial x}(x_1,y_1)- \frac{\partial \hat{k}_{\theta}}{\partial x}(x_2,y_2)\|\leqslant \kappa_2 \|(x_1,y_1)-(x_2,y_2)\| \quad \forall (x_1,y_1), (x_2,y_2)\in \left(\mathcal{X}\times\mathcal{Y}\right)^2$
        \item (A\ref{ass:learning}.4):$\| \frac{\partial f}{\partial x}(x_1)- \frac{\partial f}{\partial x}(x_2)\|\leqslant l_f \|x_1-x_2\| \quad \forall x_1,x_2 \in \mathcal{X}^2$
    \end{itemize}
\end{ass}
\noindent
In the preposition hereafter, we establish a computable learning generalization error on $\mathcal{X}\times\mathcal{Y}$.

\begin{proposition}
    \label{lem:learningbounds}
    Consider the nonlinear system in \eqref{model} satisfying Assumption \ref{asslearning}. Denote by $\hat{k}_{\theta}$ and $\hat{P}$ the correction term and contraction metric obtained by minimizing \eqref{eqoptimize} and satisfying \eqref{residual}. If Assumption \ref{ass:learning} holds, then there exists a $\mathcal{KL}$ function $\rho$, such that the learned correction term and contraction metric satisfy the following 

\begin{subequations}\label{residualonSet}
\renewcommand{\theequation}{\theparentequation.\alph{equation}}
\begin{empheq}[left=\empheqlbrace]{align}
&\mathrm{He}\left\{ \hat{P}\left[ \tfrac{\partial f}{\partial\hat{x}}(\hat{x}) +  \tfrac{\partial \hat{k}_{\theta}}{\partial\hat{x}}(\hat{x},y)\right] \right\} + c\hat{P} \leq \left(\delta+\lambda_{max}(\hat{P})\sqrt{n+p}\rho(\mu,N_v)\left( l_f+ \kappa_2 \right)\right) I,
   && \forall (\hat{x},y)\in \mathcal{X}\times\mathcal{Y} \label{res:a}\\
& \|\hat{k}_{\theta}(x, h(x))\| \leqslant \bar{\varepsilon}+\kappa_{1x}\sqrt{n}\rho(\mu,N_v), 
   && \forall x \in \mathcal{X}. \label{res:b}
\end{empheq}
\end{subequations}
  
\end{proposition}

\begin{proof}  
\underline{Learning error bound on the MPDI:}
For notational simplicity, we adopt the notation introduced in \eqref{eq:contractionInequalityV1}
\begin{equation}
 D(\hat{x},y)  \triangleq \mathrm{He}\left\{ \hat{P}\left[ \tfrac{\partial f}{\partial\hat{x}}(\hat{x}) +  \tfrac{\partial \hat{k}_{\theta}}{\partial\hat{x}}(\hat{x},y)\right] \right\} + c\hat{P}. \notag
\end{equation}
From assumption (A \ref{ass:learning}.3) and (A \ref{ass:learning}.4), one gets

\begin{align}
\|D(x_1,y_1)-D(x_2,y_2)\| &\leqslant  \left\|\hat{P} \left[ \tfrac{\partial f}{\partial x}(x_1) - \tfrac{\partial f}{\partial x}(x_2) +  \tfrac{\partial \hat{k}_{\theta}}{\partial x}(x_1,y_1)-\tfrac{\partial \hat{k}_{\theta}}{\partial x}(x_2,y_2)\right]\right\| \notag      \\
&\leqslant  \lambda_{max}(\hat{P})\left[ l_f\| x_1-x_2\|+ \kappa_2 \| (x_1,y_1)-(x_2,y_2)\| \right] \notag \\
&\leqslant  \lambda_{max}(\hat{P})\left( l_f+ \kappa_2 \right) \| (x_1,y_1)-(x_2,y_2)\|.
\label{eq:22}
\end{align}

\noindent
Relying on Proposition 2 in Appendix A, for any $(x_1,y_1)=(\hat{x},y)\in\mathcal{X}\times\mathcal{Y}$ and its nearest sample point 
$(x_2,y_2)=(\hat{x}^{[j]},y^{[j]})\in\mathcal{V}$  we have
$$\|(\hat{x},y)-(\hat{x}^{[j]},y^{[j]})\|_{\infty}\leq  \rho(\mu,N_v)\triangleq\left(\frac{\mu(\mathcal{R})}{N_v}\right)^{\frac{1}{n+p}},$$
Substituting in \eqref{eq:22}, leads to 
\begin{align}
D(\hat{x},y) &\leqslant 
D(\hat{x}^{[j]},y^{[j]}) +\lambda_{max}(\hat{P})\sqrt{n+p}\rho(\mu,N_v) \left( l_f+ \kappa_2 \right)I.
\label{eq:23}
\end{align} 

\noindent
Finally from \eqref{residual:a}, one gets
\begin{align}
D(\hat{x},y) &\leqslant \left(\delta+\lambda_{max}(\hat{P})\sqrt{n+p}\rho(\mu,N_v)\left( l_f+ \kappa_2 \right)\right) I.
\end{align} 

\noindent
\underline{Learning error bound on the BC:}
The proof follows directly from Assumption (\ref{ass:learning}.1) and (\ref{ass:learning}.2) by using the same steps as above. 

\end{proof}

\subsection{Robustness and disturbance rejection analysis}
\label{subsec:remark}
Consider now the following nonlinear system affected by a bounded process noise $w(t)$ and measurement noise $v(t)$
\begin{equation}
    \label{modelnoise}
    \left\{\begin{array}{l}\dot{x}(t)=f(x)+w(t)\\
    y_e(t)=h(x)+v(t),
    \end{array}\right.
\end{equation}
satisfying the following assumptions
\begin{ass} There exists positive constants $\bar{w}$ and  $\bar{v}$ such that
\begin{itemize}
    \item (A1):  $\norm{w(t)}\leqslant \Bar{w} \quad \forall t\in \mathbb{R}$, 
    \item (A2):    $\norm{v(t)}\leqslant \Bar{v} \quad \forall t\in \mathbb{R}$.
\end{itemize}
\label{assumvandw}
\end{ass}

\begin{ass}\label{assw}
For any bounded process noise $w\in\mathcal{W}$, the solutions of interest of system \eqref{modelnoise} initialized in $\mathcal{X}_0$, remain in the compact convex set $\mathcal{X}$, with $\mathcal{X}_0\subseteq\mathcal{X}$, and $\mathcal{X}$ as in Assumption \ref{asslearning}.
\end{ass} 

\noindent
Note that Assumption \ref{assw} is required to ensure that the process noise $w(t)$ does not drive the solutions of \eqref{modelnoise} out of the set of interest $\mathcal{X}$. As this might result in a degradation of the learning performance, and most importantly, Assumption \ref{assw} makes sure that the contraction condition still holds for the perturbed system in \eqref{modelnoise}.   
\\

\noindent
Consider now the associated observer
\begin{equation}
\label{observernoise}
\left\{\begin{array}{l}
\dot{\hat{x}}=f(\hat{x})+\hat{k}_{\theta}(\hat{x}, y_e) \\
\hat{y}=h(\hat{x}),
\end{array}\right.
\end{equation}
 where $\hat{k}_{\theta}$ is the learned correction term satisfying the following assumption:  
\begin{ass} The learned correction term $\hat{k}_{\theta}(\hat{x},y)$ is Lipschitz on its second argument, uniformly in $\hat{x}$ 
\begin{equation}
\label{lipschitz}
 \norm{ \hat{k}_\theta(., y_1)-\hat{k}_\theta(., y_2)}\leq  \kappa_y \norm{y_1-y_2},  \forall \hat{x}\in \mathcal{X}  
\end{equation}
with $\kappa>0$
\label{assumlipschitz}
\end{ass}

\begin{theorem}\label{thISS}
Consider system \eqref{modelnoise} affected with process and measurement noise satisfying Assumption \ref{assumvandw}. Furthermore, consider the observer in \eqref{observernoise} with the learned correction term $\hat{k}_{\theta}$ and contraction metric $\hat{P}$ satisfying \eqref{residualonSet}, with the correction term additionally satisfying  Assumption \ref{assumlipschitz}. There exists a computable bound $\bar{\delta}$ such that if $\delta<\bar{\delta}$, then there exist positive constants $\eta_0$, $\eta_1$, $\eta_2$, and $\eta_3$ such that the estimation error is exponentially input-to-state stable and satisfy 
\begin{equation}
\|x(t)-\hat{x}(t)\| \leqslant  \eta _0\|x(0)-\hat{x}(0)\| \exp{(-\eta_1 t)}+ \eta_2 \left(\bar{\varepsilon}+\frac{\kappa_{1x}}{N}\right)+\eta_3 \bar{w}+\eta_4 \bar{v},\\
\label{estimationerrorvandw}
\end{equation}
\noindent

\end{theorem}

\begin{proof}
Consider the estimation error $e=x-\hat{x}$ and the following Lyapunov function
$ V=\frac{1}{2}e^T\hat{P}e$. Taking the time derivative leads to
\begin{align}
\dot{V}& =e^T \hat{P} \left[f(x)+w(t)-f(\hat{x})-\hat{k}_{\theta}\left(\hat{x}, y_e\right) \right]  \notag \\ 
& =e^T \hat{P} \left[f(x)+w(t)-f(\hat{x})-\hat{k}_{\theta}\left(\hat{x}, y_e\right)+\hat{k}_{\theta} (\hat{x}, y)-\hat{k}_{\theta} (\hat{x}, y)\right] \notag \\
& =e^T \hat{P} [f(x)-f(\hat{x})-\hat{k}_{\theta}(\hat{x}, y)]  +e^T \hat{P} \left[\hat{k}_{\theta}\left(\hat{x}, y\right)-\hat{k}_{\theta}\left(\hat{x}, y_e\right)\right]+e^T \hat{P} w(t). \notag \\
&=e^T \hat{P} [f(x)-f(\hat{x})-\hat{k}_{\theta}(\hat{x}, y) +  \hat{k}_{\theta}(x, y) ]  +   e^T \hat{P} \left[\hat{k}_{\theta}\left(\hat{x}, y\right)-\hat{k}_{\theta}\left(\hat{x}, y_e\right)-\hat{k}_{\theta}(x, y) \right]+e^T \hat{P} w(t).
\label{eqw1}
\end{align}
\noindent
Using Young's Inequality, \eqref{eqw1} becomes 
\begin{align} 
\dot{V} \leqslant & e^{T} \hat{P} [f(x)-f(\hat{x})-\hat{k}_{\theta}(\hat{x}, y) +  \hat{k}_{\theta}(x, y) ]  +  e^T \hat{P}^2 e + \frac{1}{2} \left\|\hat{k}_{\theta}(x, y)\right\|^2+\frac{1}{2}\left\|\hat{k}_{\theta}(\hat{x}, y)-\hat{k}_{\theta}\left(\hat{x}, y_e\right)\right\|^2+\frac{1}{2}\|w(t)\|^2 .
\label{eqw2}
\end{align}
\noindent
Substituting \eqref{lipschitz} in \eqref{eqw2} and considering that $\hat{P}\leqslant \lambda_{max}(\hat{P}) I, $ one obtains 
\begin{align}
\dot{V}  & \leqslant e^{T} \hat{P} [f(x)-f(\hat{x})-\hat{k}_{\theta}(\hat{x}, y) +  \hat{k}_{\theta}(x, y) ] +   \lambda_{\max}(\hat{P}) e^T \hat{P} e + \frac{1}{2} \left\|\hat{k}_{\theta}(x, y)\right\|^2 + \frac{1}{2}\|w(t)\|^2  +\frac{\kappa_y^2}{2}\left\|v(t)\right\|^2.
\label{eqV3}
\end{align}
\noindent
 Considering Assumption \ref{assw} and using the fundamental theorem of calculus as in  \cite[see Theorem 4.3 p.~231] {BERNARD2022}
 
 \begin{align}
 \dot{V} \leqslant &  e^T\hat{P} \left( \int_0^1 \frac{\partial f}{\partial x}(\hat{x} + \alpha(x - \hat{x})) + \frac{\partial \hat{k}_{\theta}}{\partial x}(x + \alpha(\hat{x} - x), y) \, d\alpha \right) \times e+ \lambda_{\max}(\hat{P}) e^T \hat{P} e + \frac{1}{2} \left\|\hat{k}_{\theta}(x, y)\right\|^2 + \frac{1}{2}\|w(t)\|^2  +\frac{\kappa_y^2}{2}\left\|v(t)\right\|^2.
 \label{eq33}
\end{align}

\noindent
Using the results of Proposition 1 for $(\hat{x},y)\in \mathcal{X}\times \mathcal{Y}$, with $\mathcal{X}$ being convex, one obtains 
 
\begin{align}
 \dot{V} \leqslant &  - c e^T\hat{P}e+ \left(\delta+\lambda_{max}(\hat{P})\sqrt{n+p}\rho(\mu,N_v)\left( l_f+ \kappa_2 \right)\right) e^Te+  \lambda_{\max}(\hat{P}) e^T \hat{P} e + \frac{1}{2} \left(\bar{\varepsilon}+\sqrt{n}\rho(\mu,N_v) \kappa_{1x}\right)^2 + \frac{1}{2}\|w(t)\|^2  +\frac{\kappa_y^2}{2}\left\|v(t)\right\|^2, \notag \\
 \leqslant & - \left(c- \frac{\delta}{\lambda_{min}(\hat{P})}- \operatorname{Cond}(\hat{P})\sqrt{n+p}\rho(\mu,N_v)(l_f+\kappa_2)- \lambda_{\max}(\hat{P})\right) e^T \hat{P} e  +\frac{1}{2} \left(\bar{\varepsilon}+\sqrt{n}\rho(\mu,N_v)\kappa_{1x}\right)^2+\frac{1}{2}\|w\|^2+\frac{\kappa_y^2}{2} \|v\|^2 ,
\end{align}
\begin{align}
 \dot{V} \leqslant &- 2 \eta_1 V +\frac{1}{2} \left(\bar{\varepsilon}+\sqrt{n}\rho(\mu,N_v)\kappa_{1x}\right)^2  +\frac{1}{2}\|w\|^2+\frac{\kappa_y^2}{2} \|v\|^2, 
\end{align}
\noindent
with $\eta_1=c- \frac{\delta}{\lambda_{min}(\hat{P})}-\operatorname{Cond}(\hat{P})\sqrt{n+p}\rho(\mu,N_v)(l_f+\kappa_2)- \lambda_{\max}(\hat{P})$.\\
\noindent
Using the bounds on $w(t)$ and $v(t)$ in Assumption \ref{assumlipschitz}, and the comparison lemma 
\cite[lemma 3.4, p. 102-103 ]{khalil2002nonlinear}, one obtains the error bound in \eqref{estimationerrorvandw} with the following positive constants
\begin{equation}
  \eta_0=\sqrt{\frac{\lambda_{\max}(\hat{P})}{\lambda_{\min}(\hat{P})}}
, \eta_2=\eta_3= \frac{1}{2\sqrt{\eta_1\lambda_{\min}(\hat{P})}}, \text{ and } 
\eta_4= \frac{\kappa_y}{2\sqrt{\eta_1\lambda_{\min}(\hat{P})}}.  
\end{equation}
\noindent
Finally, the estimation error is exponentially input-to-state stable if $\delta \;<\; \underbrace{\lambda_{\min}(\hat{P})\left(
    c - \lambda_{\max}(\hat{P}) - \operatorname{Cond}(\hat{P}) \sqrt{n+p}\rho(\mu,N_v)\,(l_f + \kappa_2)
\right)}_{\,\triangleq\, \bar{\delta}} $.
\end{proof}
\begin{remark}
    Note that similar a E-ISS bounds can be derived for the non-autonomous case except $\eta_1=c-\frac{\delta}{\lambda_{min}(\hat{P})}-\lambda_{\max}(\hat{P})- \operatorname{Cond}(\hat{P})\,(l_f + \kappa_2) \sqrt{n+p}\rho(\mu,N_v)-  l_{\Phi} \operatorname{Cond}(\hat{P})$.
\end{remark}
\begin{remark}
Since $\bar{\delta}$ is a computable quantity, a learning certificate can be added at the end of training to verify whether the residual training error for MPDI in \eqref{residual} satisfies $\delta < \bar{\delta}$, or whether further iterations are required in the training to achieve this bound.
\end{remark}

\section{Robust learning-based contraction observer}
\label{sec:robust}
A closer look at the results of Theorem \ref{thu} and Theorem \ref{thISS} shows how the spectrum of the contraction metric impacts the robustness and convergence results of the proposed observer. While the contraction metric is not explicitly part of the observer structures in \eqref{observer}, \eqref{observeru}, or \eqref{observernoise}, it influences the learning of the correction term. Therefore, we aim in this section to take the contribution of the present paper one step further by proposing a robust learning-based contraction observer.
\subsection{Problem Formulation}
The results of Theorem \ref{thu} show that an ill conditioning of the contraction metric amplifies the effect of the input and might compromise the exponential convergence of the observer. Moreover, from Theorem \ref{thISS}, one can see that an ill conditioning of the contraction metric amplifies the effect of the learning generalization error on $\mathcal{X}\times \mathcal{Y}$, potentially leading to a loss of E-ISS.
To avoid this, one might consider $P=I$. However, in addition to the restrictiveness of solving \eqref{eq:ideal1} for a Euclidean metric, this choice does not provide an optimal noise and disturbance rejection, as it will be detailed below.\\

\noindent
Since the learning generalization error is an a-posteriory consequence of the training, to properly analyze the effect of the spectrum of the contraction metric on the noise and disturbance rejection, let's first consider a negligible MPDI learning validation error $\delta\rightarrow0$, $N_v\rightarrow +\infty$, and $\hat{P}\rightarrow P$. This leads to the following coefficients from Theorem \ref{thISS}
\begin{equation}
    \eta_1=c-\lambda_{\max}(P) \quad
  , \eta_2=\eta_3= \frac{1}{2\sqrt{\left(c-\lambda_{\max}(P)\right)\lambda_{\min}(P)}}, \text{ and } 
\eta_4= \frac{\kappa}{2\sqrt{\left(c-\lambda_{\max}(P)\right)\lambda_{\min}(P)}}
\label{paramrobust}
\end{equation}

\noindent
Two main observations can be made regarding the above E-ISS result:
\begin{enumerate}[(i).]
    \item The E-ISS is maintained if  $\lambda_{\max}(P)< c$.
    \item  The attenuation factors for both the noise and disturbance rely heavily on the spectrum of the contraction metric. 
\end{enumerate}

It is clear that an optimal noise and disturbance rejection is obtained by minimizing $\eta_2$, $\eta_3$, and $\eta_4$. From the parameters in \eqref{paramrobust}, minimizing  $\eta_2$, $\eta_3$ and $\eta_4$ is equivalent to maximizing $\lambda_{\min}(P)$ and minimizing $\lambda_{\max}(P)$; therefore, minimizing $\lambda_{\max}(P)-\lambda_{\min}(P)$. \\

\noindent
One might be tempted to simply consider $\lambda_{\max}(P)=\lambda_{\min}(P)=\lambda$, which, by the spectral decomposition theorem \cite[Theorem 4.1.5]{horn2012matrix}, is equivalent to $P=\lambda I$. Therefore, the optimal noise rejection is obtained by maximizing  $\left(c-\lambda\right)\lambda$ leading to  $\lambda^*=\frac{c}{2}$, i.e., $P=\frac{c}{2}I$.

\noindent
Now notice that the contraction inequality in \eqref{eq:ideal1} is scalar invariant in $P$, i.e., the solution for $P=\lambda I$ for any $\lambda>0$ is the same as for $P=I$.  Therefore, the above optimal noise rejection result is only valid for $c=2$, which is equivalent to the Euclidean case in \cite{MaraniIsrael2025}. This result entails that using a Euclidean metric (and therefore, scalar multiples of it) does not provide optimal noise and disturbance rejection for any contraction rate $c$. This further highlights the advantage of using a non-Euclidean contraction metric.\\

\noindent
Driven by the above observations, we propose a robust learning-based contraction observer by promoting the contraction conditions, the E-ISS requirements, and the optimal noise and disturbance rejection in the learning process. This is formulated into the following optimization problem: \begin{equation}\label{eq:ideal}
\left\{ \begin{aligned}
  \min_{k,\,P} \quad & \operatorname{cond}(P)-1 \\[4pt]
  \text{s.t} \quad
  & \operatorname{He}\!\left\{P\!\left[
      \tfrac{\partial f}{\partial\hat{x}}(\hat{x})
      + \tfrac{\partial k}{\partial\hat{x}}(\hat{x},y)
    \right]\right\} + cP \leq 0,
    \quad \forall(\hat{x}, y) \in \mathcal{X} \times \mathcal{Y},\\
  & k\!\left(x,\,h(x)\right) = 0,
    \quad \forall\,x \in \mathcal{X},\\
  & \lambda_{\max}(P) < c.
\end{aligned}\right.
\end{equation} 
Where the first constraint is the MPDI  enforcing contraction at rate $c$; the second is the boundary condition ensuring exact reconstruction on the output manifold; and the third
is the E-ISS condition of Theorem \ref{thISS}. Minimizing $\operatorname{cond}(P)$
aims at minimizing the noise  and disturbance attenuation factors $\eta_2$,
$\eta_3$ and $\eta_4$ in \eqref{estimationerrorvandw}.

\subsection{Promoting robustness in the learning-based contraction observer}
Similarly to section \ref{sec:pinnDes}, we propose to use Physics Informed Neural Networks to solve the problem in \eqref{eq:ideal} for the correction term and the contraction metric. Consider $\hat{k}_{\theta}$ and $\hat{P}$ as defined in \eqref{eq:mlpCompMath} and \eqref{eq:p}, respectively, be the solution to the following optimization problem
\begin{equation}
    \min_{\theta,L} \mathcal{L}_R(\theta, L)= \mu_1 \mathcal{L}_{\text{MPDI}}(\theta,L)
+ \mu_2 \mathcal{L}_{\text{BC}}(\theta)
+ \mu_3 \mathcal{L}_{P}(L)
+ \mu_4 \mathcal{L}_{iss}(L)+\mu_5 \mathcal{L}_{R}(L),
\label{lossRobust}
\end{equation}

where $\mathcal{L}_{\text{MPDI}},\,\mathcal{L}_{\text{BC}}$ and
$\mathcal{L}_{P}$ are as in~\eqref{eqoptimize}, while $\mathcal{L}_{iss}$ and $\mathcal{L}_{R}$  are related to the E-ISS result in Theorem \eqref{thISS}. The set of coefficients $\{\mu_1, \mu_2,$ ...., $\mu_5\}$ are tunable coefficients balancing the influence of each component. The three original terms preserve the contraction
property and the admissibility of $\hat{P}$, while the new terms act
as soft constraints that bias the optimization towards metrics for
which the E-ISS result  of
Theorem~\ref{thISS} is sharpest.

\subsubsection{E-ISS loss function}
From the findings of Theorem \ref{thISS}, the E-ISS is maintained if the spectral norm of $\hat{P}$ is less than the contraction rate. Therefore, we construct the following loss to ensure the E-ISS of the proposed learning-based contraction observer
\begin{equation}
\label{eq:EissLoss}
\mathcal{L}_{ISS}= \frac{1}{\beta} \ln\left(1 + \exp{\left(\beta (\lambda_{max}(\hat{P})-c)\right)}\right)
\end{equation}
 with $\beta$ being a user-set large positive constant. 

\subsubsection{Noise attenuation loss function}
An optimal noise and disturbance rejection is obtained by minimizing $\eta_2$, $\eta_3$ and $\eta_4$ in \eqref{estimationerrorvandw}, which is achieved when $\lambda_{\max}(P)$ and
$\lambda_{\min}(P)$ are close, i.e., when $\hat{P}$ is
well-conditioned. Since the MPDI does not fix
$P$ uniquely, this residual freedom can be exploited during
training to steer the learning towards metrics that are intrinsically
robust.
Recalling that the metric is parameterized through its Cholesky-like
factor as $\hat{P}=LL^\top$, the condition numbers of $L$ and
$\hat{P}$ are related by
\begin{equation}
    \operatorname{cond}(\hat{P}) \;=\; \operatorname{cond}(L)^{2},
    \label{eq:condLP}
\end{equation}
so that acting on $L$ directly controls the conditioning of $\hat{P}$.
This motivates the noise attenuation promoting loss
\begin{equation}
\mathcal{L}_{R}(L)\;=\;\operatorname{cond}(L)^{2}-1,
\label{eq:losscond}
\end{equation}
which is non-negative and scale-invariant in $L$.\\

\noindent
The training of the proposed robust learning-based contraction observer is detailed in Algorithm 2
\begin{algorithm}
\label{AlgoR}
\caption{Training Algorithm of the Robust learning-based Contraction Nonlinear Observer with Learnable Metric $P$}
\begin{algorithmic}[1]
\State \textbf{Inputs:} dataset $\mathcal{D}=\{\hat{x}^{[j]},y^{[j]}\}_{j=1}^{N}$; parameters $\Theta=(\theta,L)$; dynamics $f$, measurement $h$; Adam learning rate $\alpha$; softplus sharpness $\beta>0$; epochs $N_{\text{Adam}},N_{\text{LBFGS}}$; loss weights $(\mu_1,\dots,\mu_5)$ (set $\mu_5=0$ for the non-optimized metric, $\mu_5>0$ for the robust metric); contraction rate $c>0$; regularization constant $\epsilon$; batch size $|\mathcal{B}|$.
\For{$q=1,\dots,N_{\text{Adam}}$} \Comment{Adam phase}
  \For{mini-batch $\mathcal{B}\subset\mathcal{D}$}
    \State $P \gets L L^\top$
    \State Compute $\hat{k}_\theta(\hat{x},y)$ on $\mathcal{B}$
    \State Compute the jacobians $\big\{\frac{\partial f}{\partial \hat{x}}(\hat{x}),\,\frac{\partial \hat{k}_\theta}{\partial \hat{x}}(\hat{x},y)\big\}$ on $\mathcal{B}$
    \State Build $D(\hat{x},y) = \operatorname{He}\!\Big\{ P\big[\tfrac{\partial f}{\partial \hat{x}}+\tfrac{\partial \hat{k}_\theta}{\partial \hat{x}}\big]\Big\} + cP$
    \State $\mathcal{L}_{\text{MPDI}} \gets \frac{1}{|\mathcal{B}|}\sum_{j}\sum_{i=1}^{n}\frac{1}{\beta}\ln\!\big(1+\exp(\beta\,\lambda_i(D(\hat{x}^{[j]},y^{[j]})))\big)$ \Comment{Eq.~\eqref{eq:minorsV1}}
    \State $\mathcal{L}_{\text{BC}} \gets \frac{1}{|\mathcal{B}|}\sum_{j}\|\hat{k}_\theta(\hat{x}^{[j]},h(\hat{x}^{[j]}))\|^2$ \Comment{Eq.~\eqref{eq:eqBCV1}}
    \State $\mathcal{L}_{P} \gets \sum_{k=1}^{n}\frac{1}{l_{ii}^2+\epsilon}$ \Comment{Eq.~\eqref{eq:LdiagLossV1}}
    \State $\mathcal{L}_{\text{ISS}} \gets \frac{1}{\beta}\ln\!\big(1+\exp(\beta(\lambda_{\max}(P)-c))\big)$ \Comment{Eq.~\eqref{eq:EissLoss}}
    \State $\mathcal{L}_{R} \gets \operatorname{cond}(L)^2-1$ \Comment{Eq.~\eqref{eq:losscond}}
    \State $\mathcal{L}(\theta,L) \gets \mu_1\mathcal{L}_{\text{MPDI}}+\mu_2\mathcal{L}_{\text{BC}}+\mu_3\mathcal{L}_{P}+\mu_4\mathcal{L}_{\text{ISS}}+\mu_5\mathcal{L}_{R}$ \Comment{Eq.~\eqref{lossRobust}}
    \State $\Theta \gets \Theta - \alpha\,\nabla_{\Theta}\mathcal{L}(\theta,L)$ \Comment{jointly update $\theta$ and $L$}
  \EndFor
\EndFor
\If{$N_{\text{LBFGS}}>0$} \Comment{optional refinement}
  \For{$q=1,\dots,N_{\text{LBFGS}}$}
    \State $P \gets L L^\top$
    \State Recompute $\mathcal{L}(\theta,L)$ via Eq.~\eqref{lossRobust} (full data or large batch)
    \State $\Theta \gets \text{LBFGS\_step}\big(\Theta,\nabla_{\Theta}\mathcal{L}\big)$
  \EndFor
\EndIf
\State \textbf{Output:} optimized parameters $\Theta^\star=(\theta^\star,L^\star)$.
\end{algorithmic}
\end{algorithm}

\section{Numerical Simulations}
\label{sec:sim}

\noindent
In this section, we evaluate the performance of the proposed learning-based nonlinear contraction observer by conducting numerical simulations on four distinct nonlinear systems: the Van der Pol oscillator \eqref{VdP}, the reverse Duffing oscillator \eqref{duffing}, the Lotka-Volterra system \eqref{LV}, and the Rossler attractor \eqref{Rossler}.

\begin{center}
\begin{tabular}{p{4.7cm} p{2.5cm} p{4.0cm} p{4.5cm}}
\normalsize
 \begin{equation}\label{VdP}
\!\!\!\!\left\{\begin{array}{l}
\dot{x}_{1}\!=  x_{2} \\
\dot{x}_{2}=-x_{1}+x_2(1-x_{1}^2)\!+u \\
y=x_{1},
\end{array}\right.
\end{equation}  &  \normalsize \begin{equation}\label{duffing}
    \!\!\!\!\left\{\begin{array}{l}
    \dot{x}_{1}\!=  x_{2}^3 +u\\
    \dot{x}_{2}= -x_{1}\! \\
    y = x_{1},    
\end{array}\right.
\end{equation}
& \normalsize
\begin{equation}\label{LV}
\left\{\begin{array}{l}
\dot{x}_{1} = \frac{2}{3} x_{1} - \frac{4}{3} x_{1} x_{2}+u_1 \\
\dot{x}_{2} =  x_{1} x_{2} -  x_{2} \!+u_2 \\
y = x_{1}.
\end{array}\right.
\end{equation}
&  \normalsize \begin{equation} \label{Rossler}
\left\{\begin{array}{l}
\dot{x}_1 = -x_1 - x_2 \\
\dot{x}_2 = x_1 + 0.2 x_2 \\
\dot{x}_3 = 0.2 + x_3 (x_1 - 5.7)+u \!\\
y=x_{2}
\end{array}\right.
\end{equation}
\end{tabular}
\end{center}

\noindent
The training dataset was constructed by generating collocation points sampled uniformly from the region of interest. A total of $4,000$ samples were generated for Van der Pol oscillator from $(\mathcal{X}, \mathcal{Y})=([-2,2] \times[-3,3], [-2, 2])$ and for Reverse Duffing oscillator from $(\mathcal{X}, \mathcal{Y})=([-1,1]^2, [-1, 1])$, while $6,000$ samples were generated for the Lotka-Volterra system from $(\mathcal{X}, \mathcal{Y})=([0,3.5] \times[0,2.5], [0, 3.5])$, and for the Rossler Attractor from $(\mathcal{X}, \mathcal{Y})=([-11,11] \times[-11,11] \times [-1, 30], [-11, 11])$. Note that the Van der Pol and Reverse Duffing Oscillators can take initial conditions in $\mathcal{X}_0=\mathcal{X}$, while the Lotka-Volterra and the Rossler attractor are initialized within $\mathcal{X}_0=[1, 2]^2$  and $\mathcal{X}_0=[-1, 1]^3$, respectively. 

\noindent
For all systems, the observer correction term was parameterized by a 5-hidden-layer perceptron with 30 neurons per layer and a tanh activation function. The training employed a two-phase optimization strategy, known to be effective for physics-informed neural networks \cite{raissi2019, lou2021physics}. We first used Adam optimizer \cite{kingma2014adam} with a learning rate of $10^{-3}$ for $300$ iterations, followed by the L-BFGS algorithm \cite{byrd1995limited} for an additional $300$ iterations. The weights for the composite loss function used for each system are detailed in Table \ref{tab:weights}. All training was performed on a single NVIDIA A40 GPU.

\begin{table}[h]
    \centering
    \caption{Loss function weights used for each system.}
    \begin{tabular}{|l|c|c|c|c|c|}
    \hline
     System & $\mu_1$ & $\mu_2$ & $\mu_3$ & $\mu_4$ & $\mu_5$ \\
            & ($\mathcal{L}_{\text{MPDI}}$) & ($\mathcal{L}_{\text{BC}}$) & ($\mathcal{L}_{P}$) & ($\mathcal{L}_{\text{ISS}}$) & ($\mathcal{L}_{R}$) \\
    \hline
     Van der Pol        & $10^{-2}$ & $1$       & $10^{-4}$ & $10^{-3}$ & $10^{-2}$ \\
     Reverse Duffing    & $10^{-1}$ & $1$       & $10^{-4}$ & $10^{-3}$ & $10^{-2}$ \\
     Lotka--Volterra    & $10^{-1}$ & $10^{-1}$ & $10^{-4}$ & $10^{-3}$ & $10^{-2}$ \\
     Rossler Attractor  & $1$       & $10^{-1}$ & $10^{-4}$ & $10^{-3}$ & $10^{-1}$ \\
    \hline
    \end{tabular}
    \label{tab:weights}
\end{table}

\begin{figure}[t]
\centering
    \includegraphics[width=0.4\linewidth]{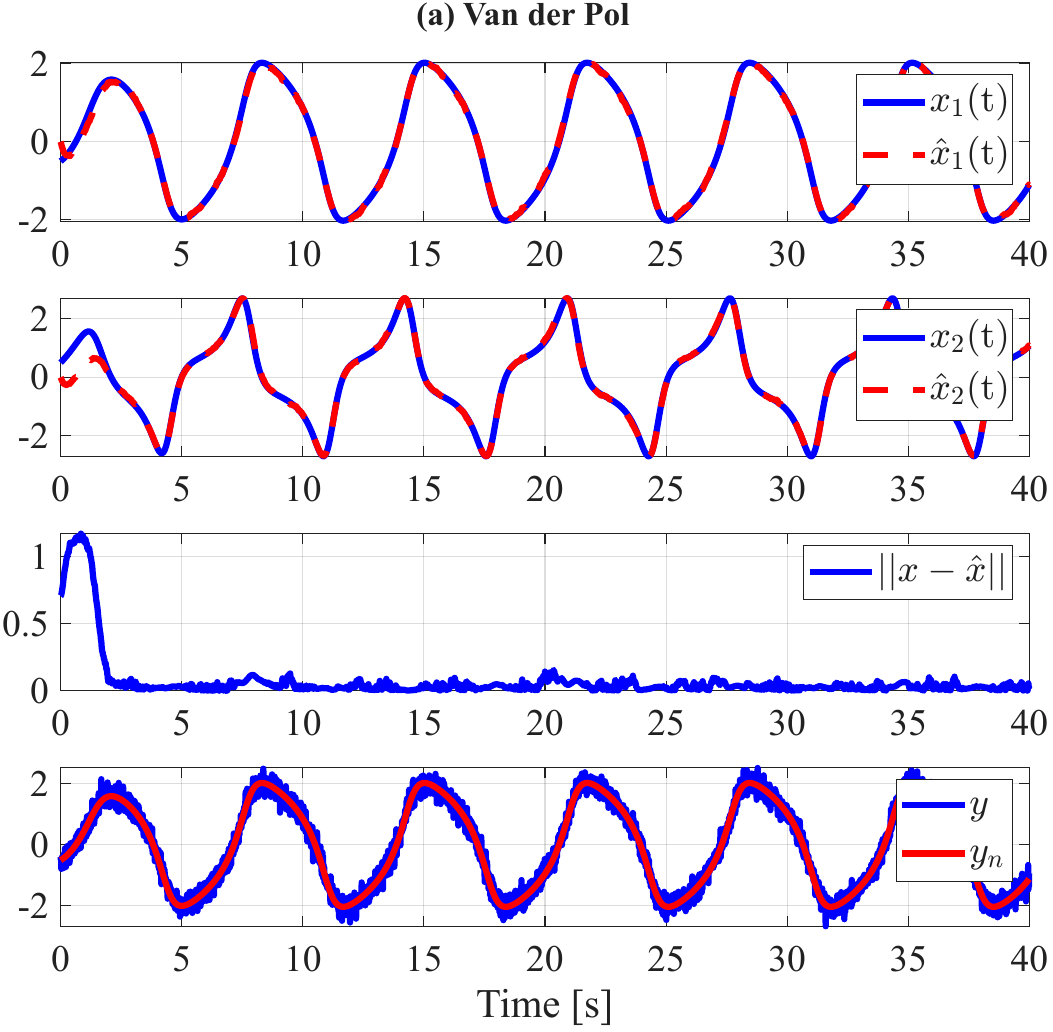} \includegraphics[width=0.4\linewidth]{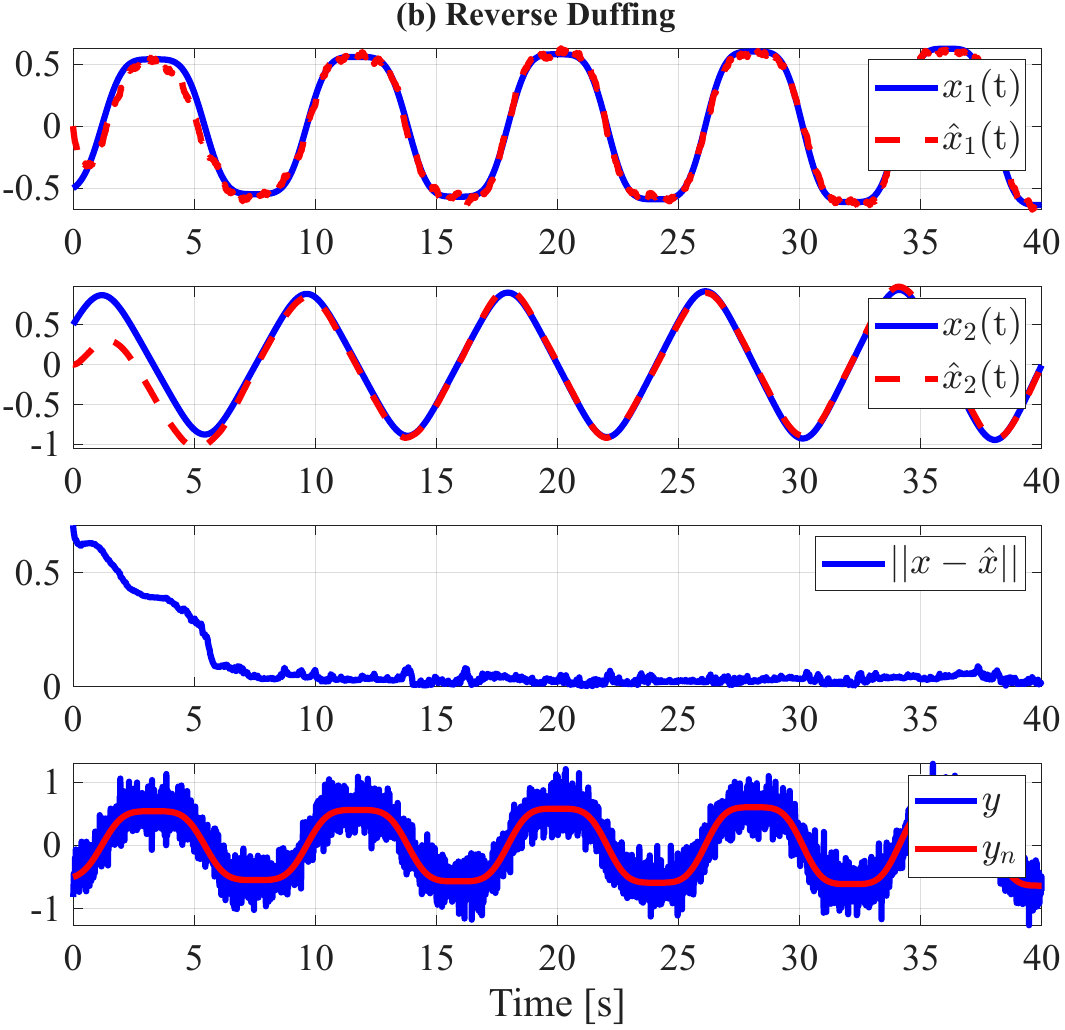} \\
    \vspace{0.5cm}
    \includegraphics[width=0.4\linewidth]{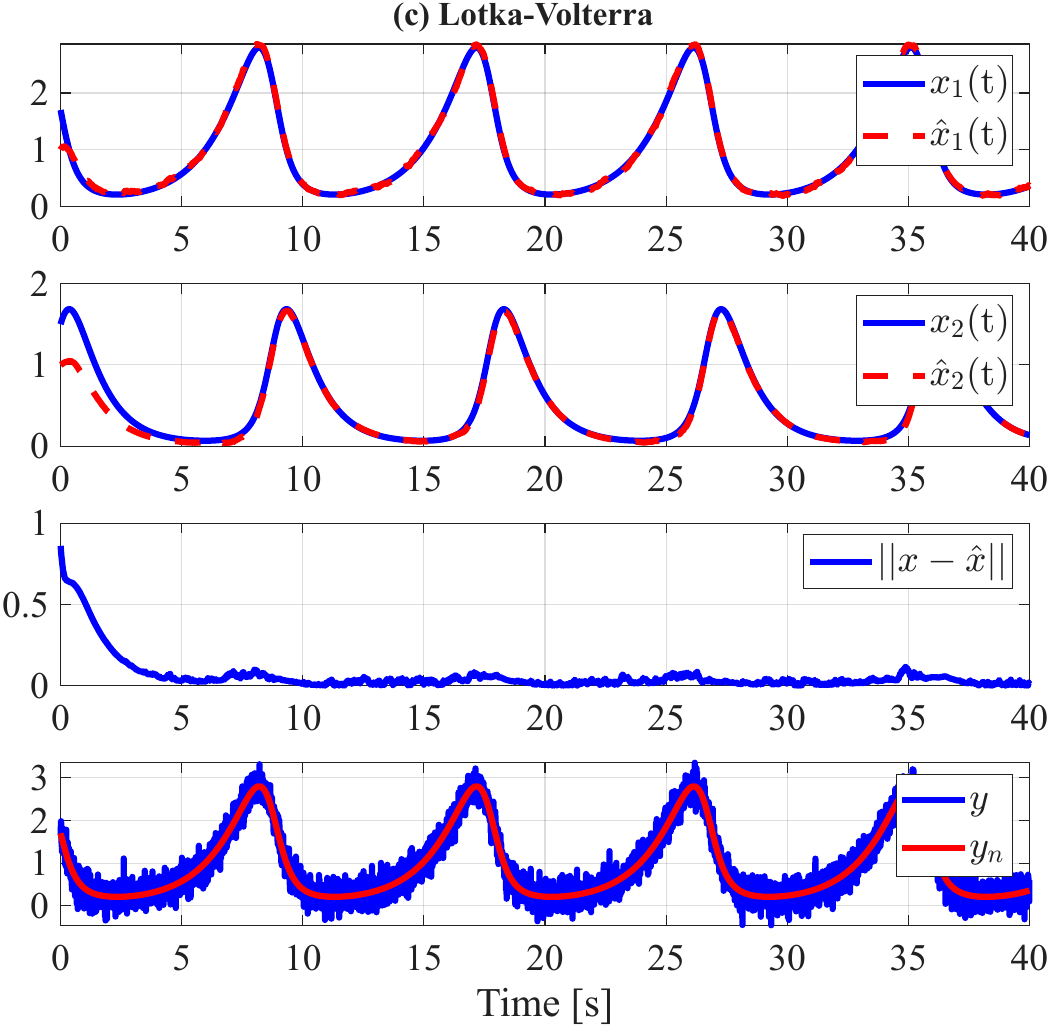}     \includegraphics[width=0.4\linewidth]{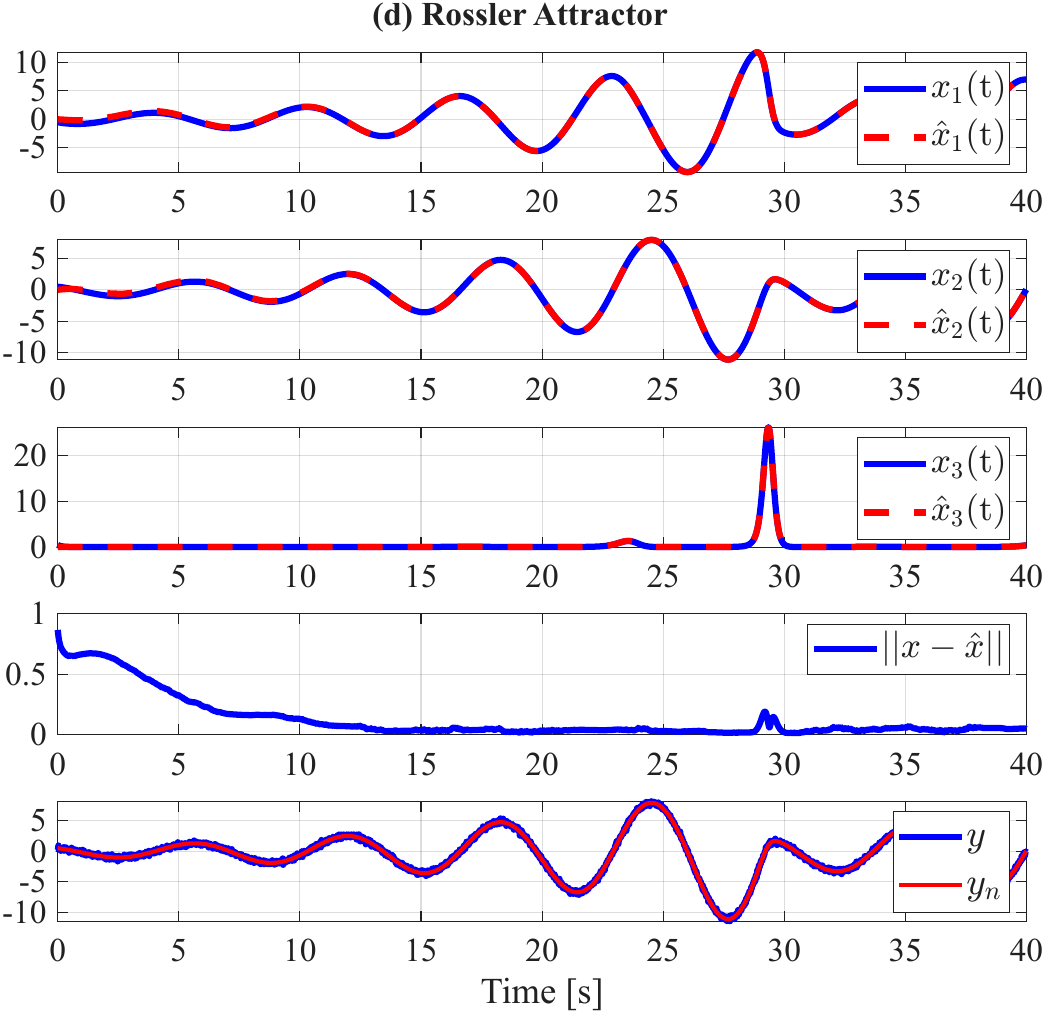} 
    \caption{Estimated states and estimation error of (a): Van der Pol oscillator, (b) Reverse Duffing, (c): Lotka-Volterra, and (d) Rossler Attractor, using the proposed robust learning-based contraction observer for $c=3$, under $20\%$ level of measurement noise. }
    \label{figestimation}
\end{figure}

\begin{figure}[t]
\centering
    \includegraphics[width=0.4\linewidth]{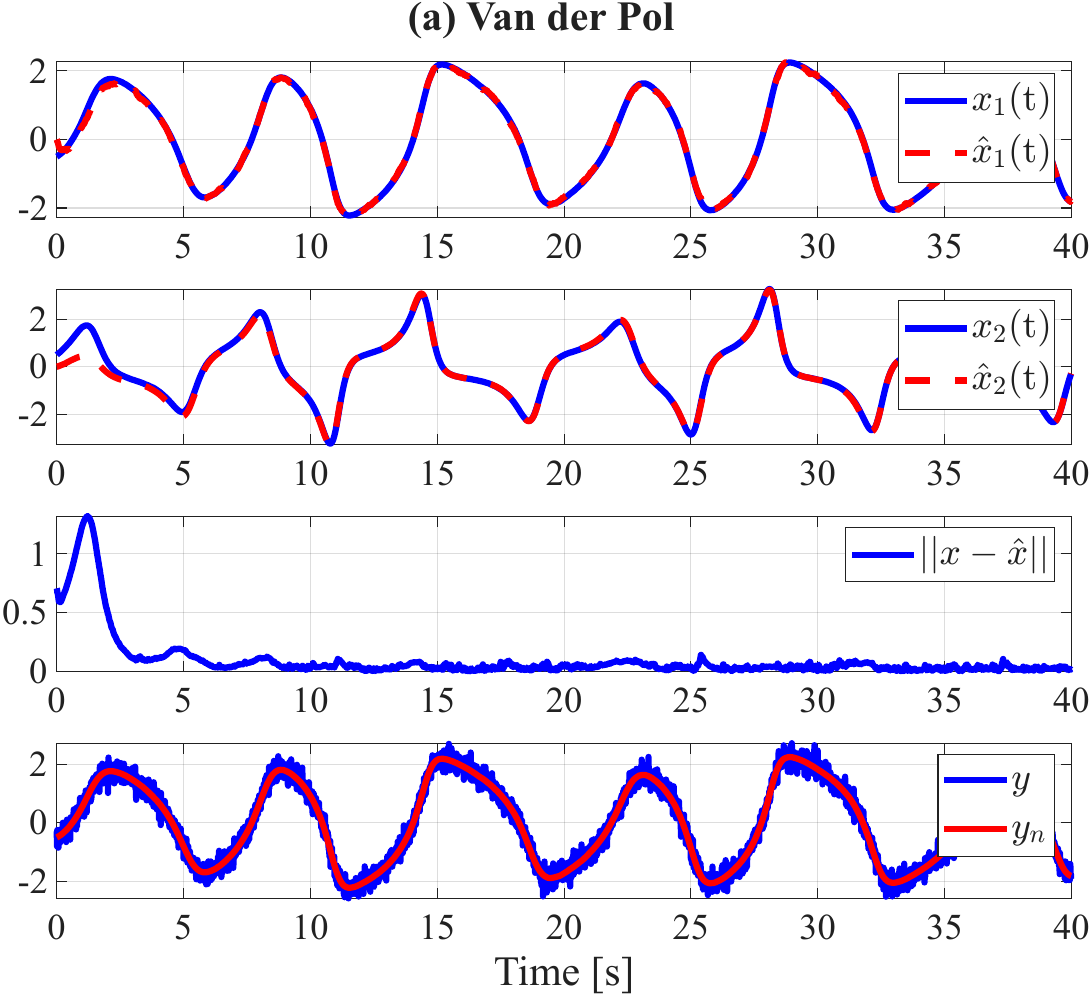} \includegraphics[width=0.4\linewidth]{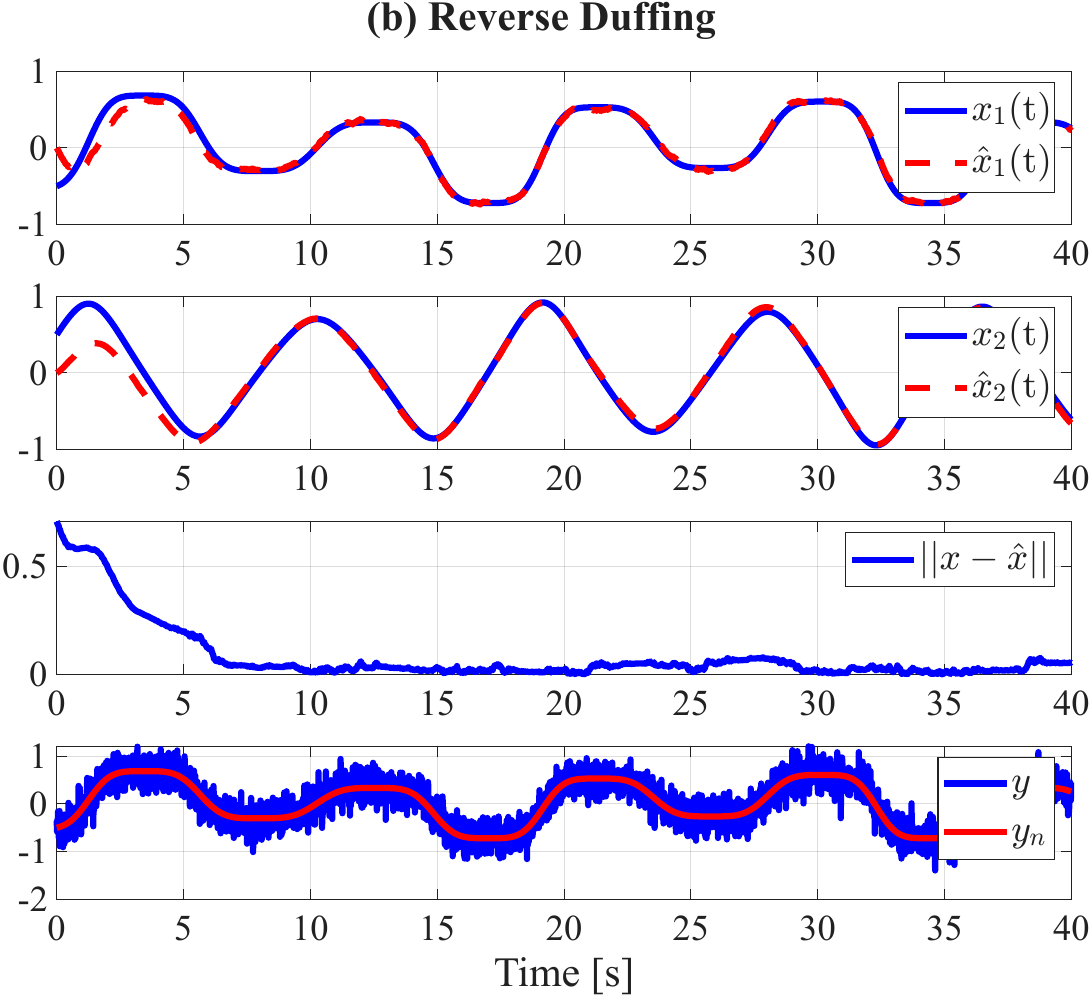} \\
    \vspace{0.5cm}
    \includegraphics[width=0.4\linewidth]{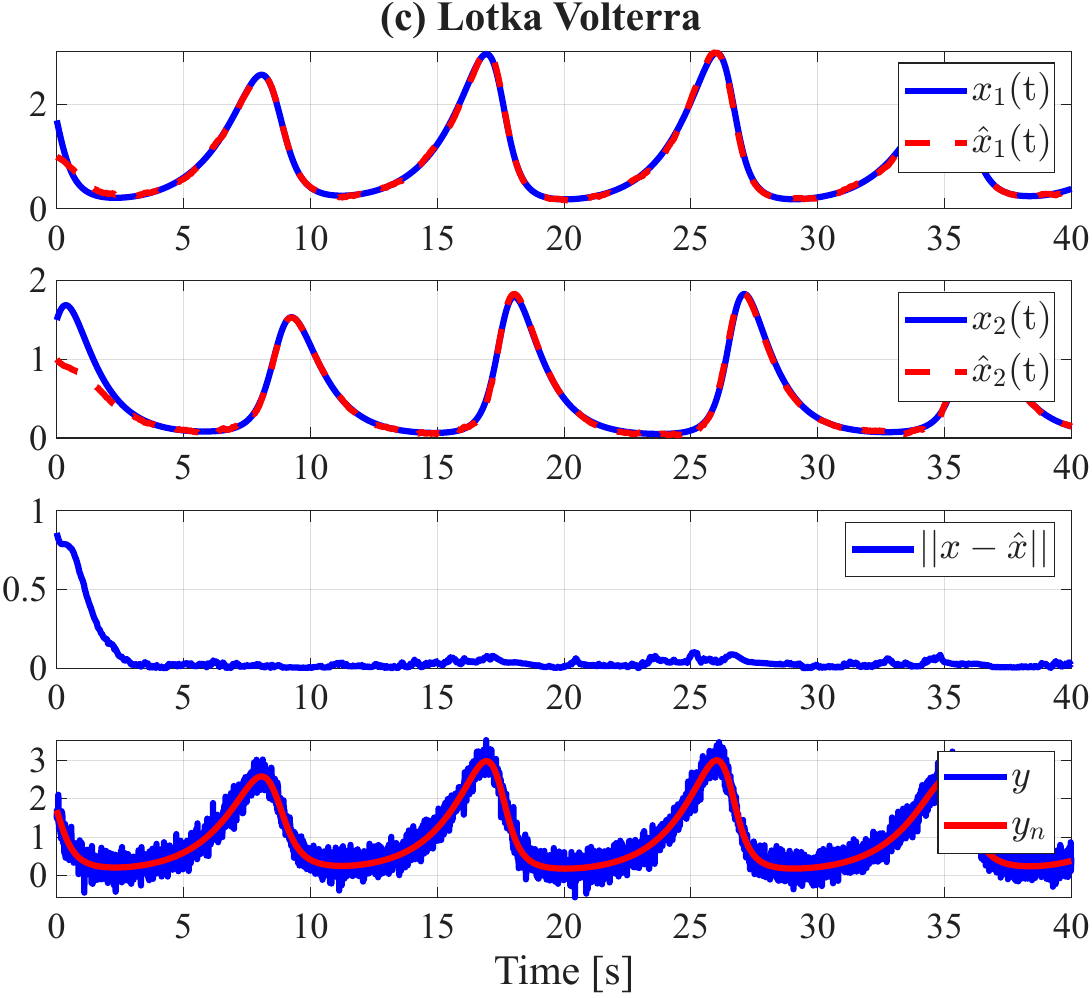}     \includegraphics[width=0.4\linewidth]{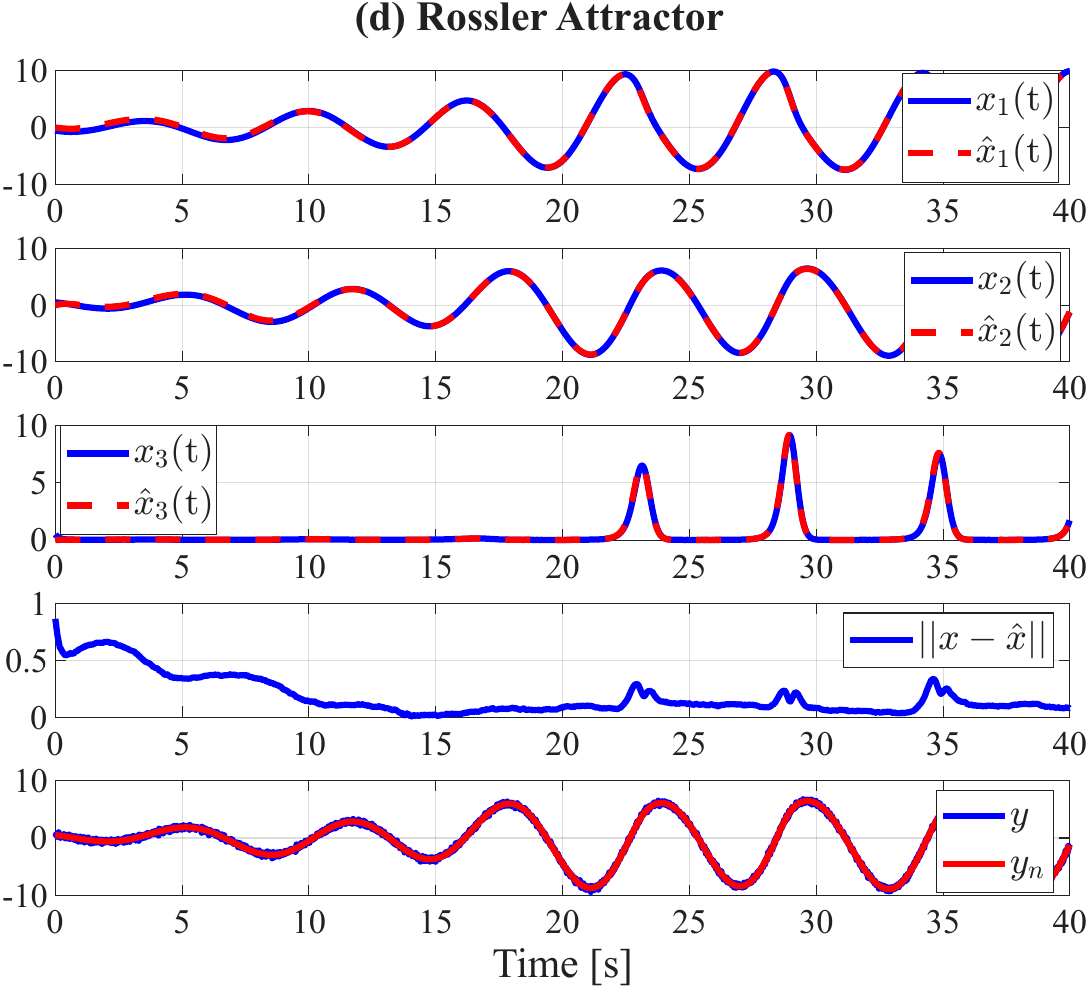} 
    \caption{Estimated states and estimation error of the forced  (a): Van der Pol oscillator, (b) Reverse Duffing, (c): Lotka-Volterra, and (d) Rossler Attractor, using the proposed robust learning-based contraction observer for $c=3$, under $20\%$ level of measurement noise. }
    \label{figestimationinput}
\end{figure}

\noindent 
FIGURE \ref{figestimation} illustrates the proposed observer's performance for the above unforced systems with a contraction rate of $c=3$ and a white Gaussian measurement noise with 0.2 standard deviation. Furthermore, a zero-mean white Gaussian process noise with 0.1 standard deviation was considered for each system.
The observer for all systems was initialized at zero, with the exception of the observer for the Lotka-Volterra system initialized at $[1, 1]$. The following initial conditions were considered for each system: $[0.5, 0.5]$ for \eqref{VdP} and \eqref{duffing}, $[1.7, 1.5]$ for \eqref{LV}, and $[0.5, 0.5, 0.5]$ for \eqref{Rossler}. Despite the presence of process and measurement noise, the estimated states converge to the true states of the system, and the estimation error converges exponentially to a small neighborhood of the origin determined by the level of measurement noise and the approximation error of the neural network, demonstrating, therefore, the effectiveness and robustness of the proposed approach.\\

\noindent
To evaluate the performance of the proposed observer for non-autonomous nonlinear systems, we consider the following inputs for each system: $u(t)=0.5\sin{(0.5t)}$ for the Van der Pol Oscillator, $u(t)=0.1\sin{(0.3t)}$ for the Reverse Duffing, $u_1(t)=0.01\sin{(0.3t)}$ and $u_2(t)=0.01\cos{(0.2t)}$ for the Lotka Volterra system, and $u(t)=0.5\sin{(0.5t)}$ for the Rossler Attractor.   
FIGURE \ref{figestimationinput} shows the estimated states of the non-autonomous systems using the proposed learning-based contraction observer with the static correction term and contraction metric under similar process and measurement noise as in FIGURE \ref{figestimation}. One can see that the 
Learning-based contraction observer learned from the autonomous version of systems \eqref{VdP}-\eqref{Rossler} can accurately estimate the states of the forced systems even in the presence of measurement noise, and the estimation error converges to the vicinity of the origin determined by the level of noise and the learning errors.

\begin{figure}[!t] 
        \centering
      \begin{overpic}[width=0.24\linewidth]{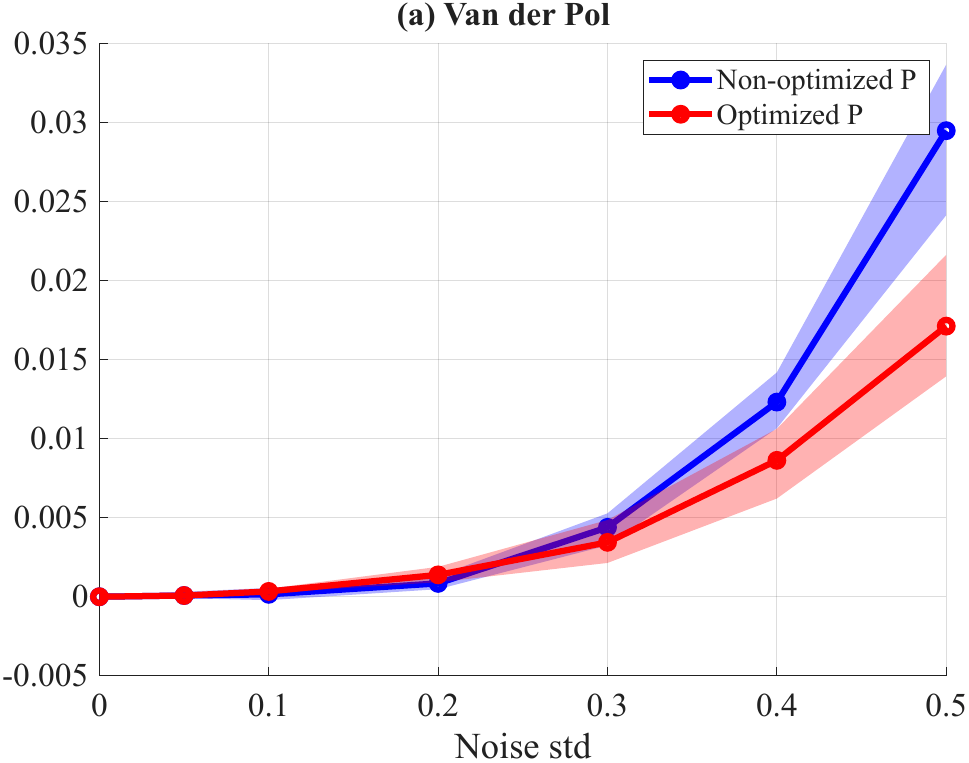}
            \end{overpic}  
      \begin{overpic}[width=0.24\linewidth]{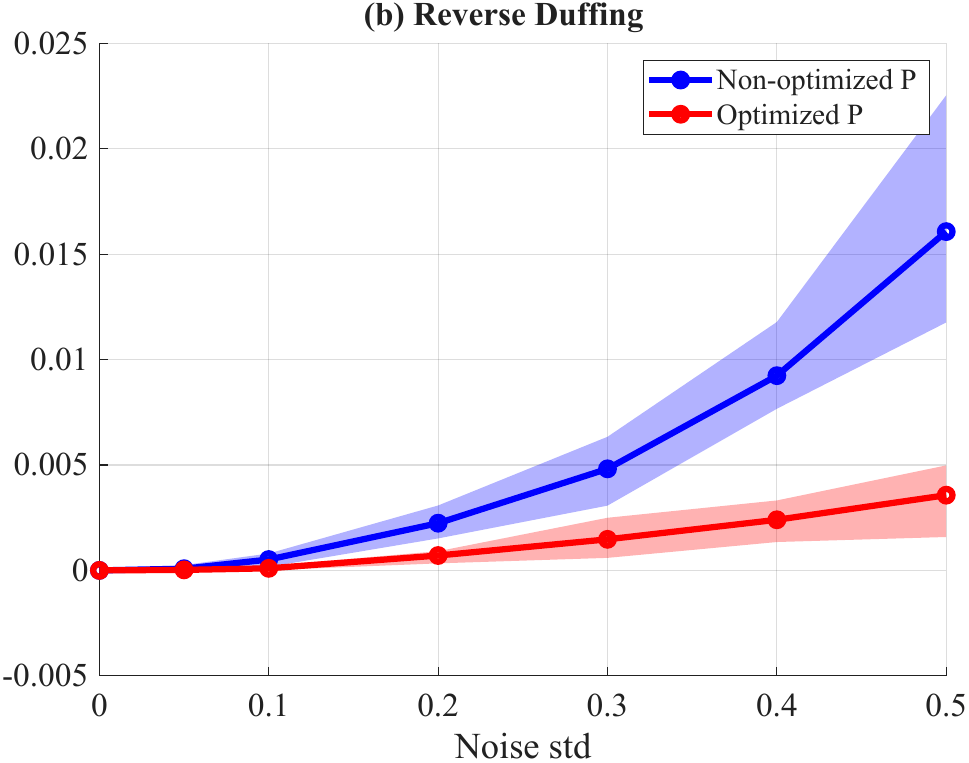}
            \end{overpic}       
     \begin{overpic}[width=0.23\linewidth]{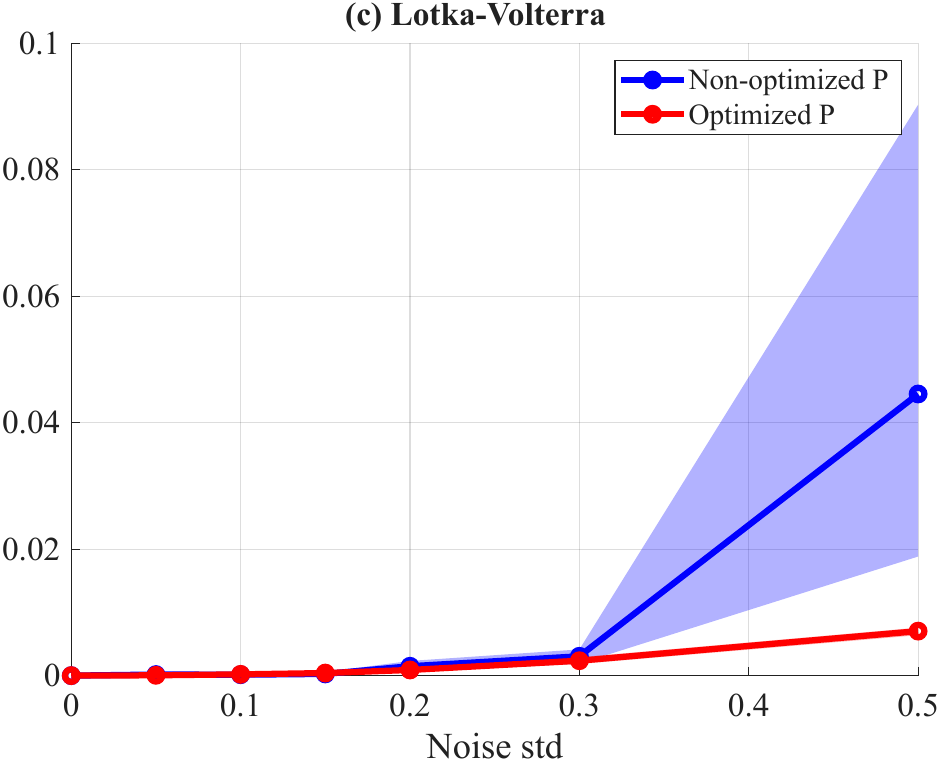}
     \put(19,35){\begin{overpic}[scale=0.12]{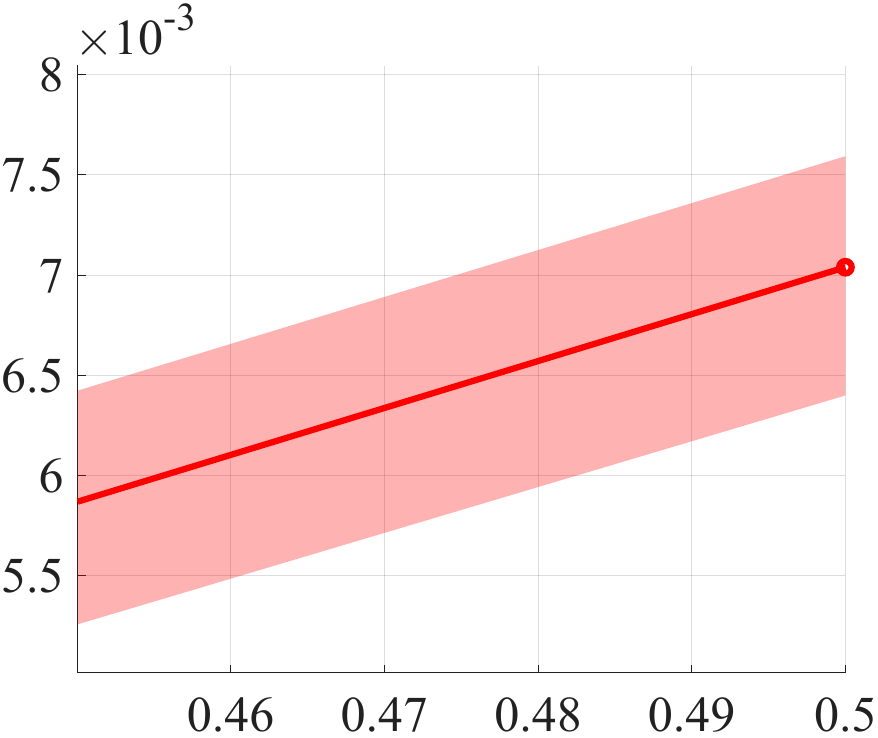}
      \end{overpic}}
     \end{overpic}
    \begin{overpic}[width=0.23\linewidth]{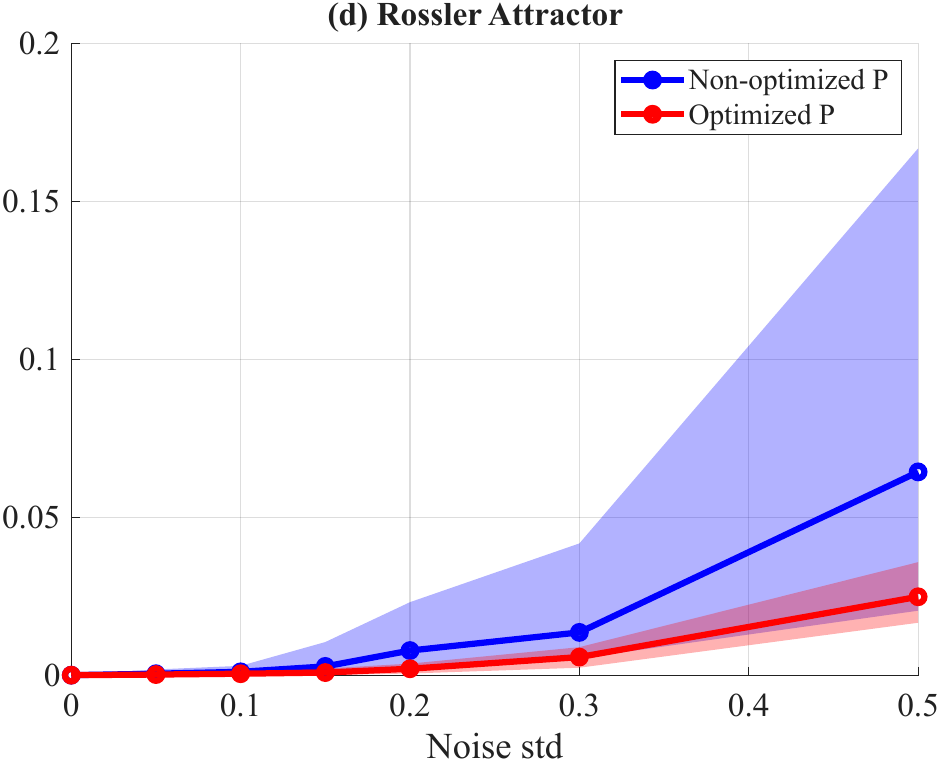}
     \end{overpic}
      \caption{The mean and standard deviation of $\Delta \text{MSE}$ over 10 different simulation scenarios for different levels of Measurement of noise in the case of optimized and non-optimized contraction metric with $c=3$.}\label{fig:mseResults}
\end{figure}

\begin{table}[h!]
\centering
\caption{Condition number of the non-optimized and optimized contraction metric P for all numerical examples.}
\label{tab:condP}

\begin{tabular}{|l| p{1cm} |p{2.5cm}| p{2.5cm}| }
\hline
{System} & {$c$} & {Non-optimized P} & {Optimized P}  \\
\hline

 & 2   & 235.7972  & 1.0124 \\
 & 2.5 & 3.8265 & 1.0086  \\
Van der Pol  & 3   & 86.5393 & 1.0886\\
 & 4   & 1.8175 & 1.3517  \\
 & 5   & 173.4178 & 1.6903 \\
\hline
 & 2   & 9.7124  & 1.0001  \\
 & 2.5 & 5749.6904 & 1.0177\\
Reverse Duffing & 3   & 15.8951 & 1.0027 \\
 & 4   & 766.9978 & 1.0054  \\
 & 5   & 2.0012 & 1.0077 \\
\hline

 & 2   & 2.2339  &1.2596 \\
 & 2.5 & 2.6063 & 1.6439  \\
Lotka–Volterra & 3   & 481.4781 & 1.7147 \\
 & 4   & 38.9651 & 1.4899  \\
 & 5   & 2.0622 & 1.4813  \\

 \hline

 & 2   & 15.3967  & 1.5039 \\
 & 2.5 & 37.7093 & 1.4587  \\
Rossler attractor & 3   & 22.3611 &  1.7375 \\
 & 4   & 53.3596 & 1.7212  \\
 & 5   & 107.3236 & 1.7746  \\
\hline
\end{tabular}
\end{table}

\noindent
To quantitatively assess the observer's robustness, we varied the level of measurement noise for a contraction rate $c=3$ for both the optimized and non-optimized contraction, i.e., by using the loss in \eqref{lossRobust} and \eqref{eqoptimize}, respectively. The measurement noise was modeled as zero-mean Gaussian white noise with a standard deviation $v_{\text{STD}} \in \{0, 0.05, 0.1,0.2, 0.3, 0.5\}$. To properly visualize the estimation errors under these scenarios independently from the learning errors, we computed the change in mean squared error $\Delta \text{MSE} (v_{\text{STD}})$ relative to the no-noise baseline
$$
\Delta \text{MSE}(v_{\text{STD}})= \text{MSE}(v_{\text{STD}})-\text{MSE}(0).
$$

\noindent
FIGURE \ref{fig:mseResults} presents the change in mean squared error $\Delta \text{MSE}(v_{\text{STD}})$ relative to the $0\%$ level of noise case. Overall, the results for both the optimized and non-optimized contraction metric indicate very little increase in the MSE in the presence of increasing levels of measurement noise. Furthermore, FIGURE \ref{fig:mseResults} highlights the significant impact of optimizing the eigenvalues of the contraction metric during the learning process. One can see that the average $\Delta \text{MSE}(v_{\text{STD}})$ over the different simulation scenarios for the four systems considered is significantly lower when the eigenvalues of the contraction metric are optimized, especially for high levels of measurement noise. Furthermore, it can be seen in FIGURE \ref{fig:mseResults} that optimizing the contraction metric notably reduces the standard deviation of $\Delta \text{MSE}(v_{\text{STD}})$.     \\
Moreover, TABLE \ref{tab:condP} illustrates the condition number of the non-optimized and optimized contraction metric P obtained using the loss functions in \eqref{eqoptimize} and \eqref{lossRobust}, respectively, for different contraction rates. It can be seen that the condition number of the contraction metric obtained through \eqref{eqoptimize} is arbitrary and can vary from small to large values, while the condition numbers of the optimized P obtained by minimizing \eqref{lossRobust} are all closer to one, leading to better robustness properties of the observer, as illustrated in FIGURE \ref{fig:mseResults}.

\section{Conclusion}\label{sec:conclusion}
This paper addresses a long-standing challenge in contraction-based nonlinear observer design by proposing an unsupervised learning-based approach for simultaneously designing the nonlinear correction term and the contraction metric. The method relies on a physics-informed neural network that embeds the contraction conditions into the learning process required to guarantee exponential convergence of the state estimation. Additionally, a learning-based contraction observer for non-autonomous systems is proposed using a static correction term and contraction metric, avoiding time-dependent PINNs that limit generalization beyond the training horizon and compromise online estimation. Computable bounds on the learning approximation error of the proposed observer were derived. The impact of the learning errors, measurement noise, and process noise on the estimation error was analyzed, leading to conditions ensuring exponential input-to-state stability.\\
Based on the robustness analysis, a robust learning-based contraction observer was proposed by optimizing the spectrum of the contraction metric during training. Numerical simulations demonstrated strong performance for both autonomous and non-autonomous systems under process and measurement noise, with effective noise rejection across a range of noise magnitudes. Future work will focus on refining learning error bounds and relaxing the compactness assumption on the state set to extend the approach to systems with unbounded state trajectories. \\

\bmsection*{Acknowledgments}
\noindent Research reported in this publication was supported by King Abdullah University of Science and Technology (KAUST) with the  Base Research Fund (BAS/1/1627-01-01) and (BAS/1/1665-01-01), and the National Institute for Research in Digital Science and Technology (INRIA). The authors would also like to thank  Ibrahima N’Doye, Filippo Fabiani, and Renzo Caballero for the insightful discussions.

\bmsection*{Financial disclosure}

None reported.

\bmsection*{Conflict of interest}

The authors declare no potential conflict of interest.

\bibliography{wileyNJD-MPS}

\section*{Appendix A: Technical Definitions}

\begin{definition}[Axis-Aligned Bounding Orthotope]
\label{defA1}
Let $\Omega \subset \mathbb{R}^{d}$ be a non empty and compact set. Its smallest enclosing Axis-Aligned Bounding Orthotope (AABO) is
\[
  \mathcal{R} = \prod_{i=1}^{d}\left[\, \min_{z \in \Omega} z_i,\ \max_{z \in \Omega} z_i \,\right],
\]
with side lengths and center, for $i = 1, \dots, d$,
\[
  s_i = \max_{z \in \Omega} z_i - \min_{z \in \Omega} z_i,
  \qquad
  (c_{\mathcal{R}})_i = \frac{\max_{z \in \Omega} z_i + \min_{z \in \Omega} z_i}{2}.
\]
\end{definition}

\begin{definition}[N-Uniform Hypercubic Cover]
\label{defA2}
For $q > 0$ and $c \in \mathbb{R}^{d}$, the hypercube of side $q$ and center $c$ is
\[
  C(q, c) = \prod_{i=1}^{d}\left[\, c_i - \tfrac{q}{2},\ c_i + \tfrac{q}{2} \,\right].
\]
Given $\mathcal{R}$ as in Definition~1 and $N \in \mathbb{N}$, set
\[
  q = \left( \frac{\prod_{i=1}^{d} s_i}{N} \right)^{\!1/d},
  \qquad
  n_i = \left\lceil \frac{s_i}{q} \right\rceil,
  \qquad
  \bar{n} = \prod_{i=1}^{d} n_i.
\]
Denote by $\mathcal{R}_N$ the N-Uniform Hypercubic Cover (N-UHC) enclosing
$\mathcal{R}$ and centered at $c_{\mathcal{R}}$, with side lengths
$\bar{s}_i = n_i\, q$, defined as
\[
  \mathcal{R}_N = \bigcup_{k \in K} C(q, c_k),
  \qquad
  K = \{1, \dots, n_1\} \times \cdots \times \{1, \dots, n_d\},
  \quad |K| = \bar{n},
\]
where each elementary hypercube center $c_k$ is given component-wise, for $i = 1, \dots, d$, by
\[
  (c_k)_i = (c_{\mathcal{R}})_i + q\left( k_i - \frac{n_i + 1}{2} \right).
\]

\end{definition}

\begin{definition}[Equidistant deterministic sample set]
\label{defA3}
Let $\mathcal{R}$ be the smallest AABO  enclosing $\Omega$, and $\mathcal{R}_N$ its N-UHC, as in Definition 1 and 2. The
\emph{equidistant deterministic sample set} of $\Omega$ is the set of the elementary hypercube
centers of $\mathcal{R}_N$ contained in $\Omega$, i.e., 
\[
  \mathcal{D} = \{\, c_k : k \in K| \ c_k \in \Omega \,\}.
\]
Its points lie on a uniform Cartesian grid of resolution 
$q=\left(\frac{\mu(\mathcal{R})}{N}\right)^{1/d}$. \\

\end{definition}

\begin{proposition}
    Let $\mathcal{D}$ be the equidistant deterministic sample set of the set $\Omega$ with resolution $q=\left(\frac{\mu(\mathcal{R})}{N}\right)^{1/d}$, for a given $N$, and $\mathcal{R}$ being the smallest AABO enclosing $\Omega$. Then $\forall z\in \Omega$ and its closest neighbor $z^{[j]} \in \mathcal{D}$ satisfy 
    \begin{enumerate}[(i).]
        \item $\|z-z^{[j]}\|_{\infty}\leqslant q.$\\
        \item  Furthermore, if $\Omega$ is an orthotope, then $|\mathcal{D}|=\bar{n}$ and $\|z-z^{[j]}\|_{\infty}\leqslant \frac{q}{2}.$
    \end{enumerate}
    
\end{proposition}

\begin{proof}
The proof of \textit{(i)} is trivial by construction of $\mathcal{D}$ and $\mathcal{R}_N$. To provide some insight, in property (i), if the closest neighbor $z^{[j]}$ to a given $z\in \Omega$ is in $\mathcal{D}$, then the distance $\|z-z^{[j]}\|_{\infty}\leqslant \frac{q}{2}.$ However, for any $\Omega$, we might be in the case where the closest neighbor to $z$ is in $\mathcal{R}_N \setminus \mathcal{D}$. Therefore, in this case, by construction in Definitions 2 and 3, the closest neighbor to $z$ in $\mathcal{D}$ satisfies $\|z-z^{[j]}\|_{\infty}\leqslant q.$
\\

\noindent
The property \textit{(ii)} means that when $\Omega$ is an orthotope, then all the centers of the elementary hypercubes of $\mathcal{R}_N$ are in $\Omega$ and therefore in $\mathcal{D}$, leading to a tighter distance than (i). To prove this, it suffices to show that the farthest elementary hypercube centers from $c_{\mathcal{R}}$ are in  $\Omega$.  \\

\noindent
If $\Omega$ is an orthotope then $\mathcal{R}=\Omega$. From Definition~2, the centers coordinate along axis $i$ are
$(c_k)_i = (c_{\mathcal{R}})_i + q\big(k_i - \tfrac{n_i+1}{2}\big)$ with
$k_i \in \{1,\dots,n_i\}$ and $i=1,....,d$, so the coordinates of the farthest hypercube centers along axis $i$  
$(c_{\mathcal{R}})_i$ are attained at $k_i = 1$ or $k_i = n_i$ and satisfy 
\begin{align}
      \big|(c_k)_i - (c_{\mathcal{R}})_i\big|
  &\le q\,\frac{n_i - 1}{2}. \notag
\end{align}

 \noindent 
Given that  $n_i = \lceil \frac{s_i}{q} \rceil$, this yields
\begin{align}
      \big|(c_k)_i - (c_{\mathcal{R}})_i\big|
  &\le \frac{s_i}{2}. \notag
\end{align}
\noindent
Hence, all $\bar{n}$ centers
$c_k$ lie strictly inside $\Omega$,  giving $\lvert \mathcal{D} \rvert = \bar{n}$.

\end{proof}





\end{document}